\RequirePackage{fix-cm}
\documentclass[smallextended]{svjour3}       
\smartqed

\usepackage{graphicx}
\usepackage{times}
\usepackage{mathptmx}
 \usepackage[numbers, sort&compress]{natbib}
 \usepackage{subcaption,graphicx}
 \usepackage{multirow}
 \usepackage{array}
\usepackage{rotfloat}
\usepackage{multicol}
\usepackage{tabularx,ragged2e}
\usepackage{microtype}
\usepackage{tikz}
\usepackage{pgfplots}
\pgfplotsset{compat=1.7}
\usepackage{subcaption}
\usepackage[mathletters]{ucs}
\usepackage[utf8x]{inputenc}
\usepackage[linesnumbered,ruled,vlined]{algorithm2e}
\usepackage{amsmath}
\usepackage{caption,float}
\usepackage{amsfonts}
\usepackage{tikz}
\usepackage{pgfplots}
\usepackage{subcaption}
\usepackage[mathletters]{ucs}
\usepackage[utf8x]{inputenc}
\usepackage[linesnumbered,ruled,vlined]{algorithm2e}
\usepackage{comment}
\usepackage{booktabs} 
\usepackage{amsmath}
\usepackage{array, tabularx, ltablex, caption}
\usepackage{enumitem}
\usepackage[sort&compress]{natbib}
\usepackage{adjustbox}
\usepackage{color}
\usepackage{amsthm}

\usepackage{url}
\DeclareMathOperator*{\argmax}{argmax} 
\begin{document}
\title{A Differential Approach for Data and Classification Service based Privacy-Preserving Machine Learning Model in Cloud Environment}
\titlerunning{New Generation Computing}  
\author{Rishabh Gupta \and
Ashutosh Kumar Singh}
\authorrunning{New Generation Computing}
\institute{Department of Computer Applications, National Institute of Technology, Kurukshetra, Haryana, INDIA\\
\email{rishabhgpt66@gmail.com}\\
\email{ashutosh@nitkkr.ac.in}\\
This article has been accepted in Springer New Generation Computing Journal \(\copyright\)  2022 Springer Nature. Personal use of this material is permitted. Permission from Springer must be obtained for all other uses in any current or future media, including reprinting/republishing this material for advertising or promotional purposes, creating new collective works, for resale or redistribution to servers or lists, or reuse of any copyrighted component of this work in other works. This work is freely available for survey and citation.
}
\maketitle
\begin{abstract}
The massive upsurge in computational and storage has driven the local data and machine learning applications to the cloud environment. The owners may not fully trust the cloud environment as it is managed by third parties. However, maintaining privacy while sharing data and the classifier with several stakeholders is a critical challenge. This paper proposes a novel model based on differential privacy and machine learning approaches that enable multiple owners to share their data for utilization and the classifier to render classification services for users in the cloud environment. To process owners’ data and classifier, the model specifies a communication protocol among various untrustworthy parties. The proposed model also provides a robust mechanism to preserve the privacy of data and the classifier. The experiments are conducted for a Naive Bayes classifier over numerous datasets to compute the proposed model's efficiency. The achieved results demonstrate that the proposed model has high accuracy, precision, recall, and F1-score up to 94\%, 95\%, 94\%, and 94\%, and improvement up to 16.95\%, 20.16\%, 16.95\%, and 23.33\%, respectively, compared with state-of-the-art works.
\end{abstract}
\section{INTRODUCTION}
Cloud computing has the capability to provide an ample amount of data storage, computation, analysis, and sharing services to organizations without revealing its implementation and platform details \cite{iwasa2020development}. Due to these illimitable services, any organization outsources its data and model from the local to the cloud platform \cite{xu2018secure}, \cite{stergiou2017efficient}. Nowadays, machine learning has gained a lot of attention in a wide range of real-world applications, including image recognition, spam detection, financial market analysis, and recommendation systems \cite{sarker2021machine}, \cite{batouche2021comprehensive}. The machine learning classifier is extensively utilized in these applications. Thereafter, the emerging challenges include handling the privacy of the model and sensitive data \cite{ghorbel2017privacy}. The cloud stores and computes massive amounts of collected data and the classification model that are outsourced from various owners without interaction with each other. The owners hesitate to share their data and model with the cloud for computation and storage since a third party runs it \cite{ali2015security}. The owners lose control of their outsourced data, model and are unaware of its extraction from the cloud platform \cite{shen2018enabling}. 
According to a Cisco survey, 76\% of owners have no clue about their data utilization by other parties \cite{cisco}.
The cloud may misuse and provide both data and the model to other parties for different purposes \cite{wei2014security}, \cite{peyvandi2021computer}. Due to these reasons, the privacy-preserving of data and the model has become a challenge for any organization. Therefore, it is essential to protect the data and classification model by applying some privacy strengthening process before transferring them to an untrusted cloud server.
\par
To address the challenges mentioned above, we propose a Differential Approach for data and classification service-based Privacy-preserving Machine Learning Model (DA-PMLM) in the cloud environment. In the proposed model, $\epsilon$-differential privacy protection is considered on the owner's side because they do not aspire to disclose actual data, and the model \cite{dwork2006calibrating, zhao2019novel, fang2020local,thilakanathan2014secure}. Fig. 1 outlines a bird-eye view of the proposed work and emphasizes our consistent contributions to protect the data as well as the classification model from unauthorized parties. The owners inject different statistical noise into data and the model according to various applications and queries. The achieved data and model are uploaded to the cloud platform, and classification services are provided. The obtained machine learning model is applied over collected data for classification. The cloud platform receives the classified data from the classification model and distributes it to the data owners or users rather than other parties. It classifies tasks and allows data sharing in the cloud environment. The main contributions of the proposed work have been described as follows:
 	\begin{figure*}[!htbp]
	\begin{center}
	\includegraphics[width=0.75\textwidth]{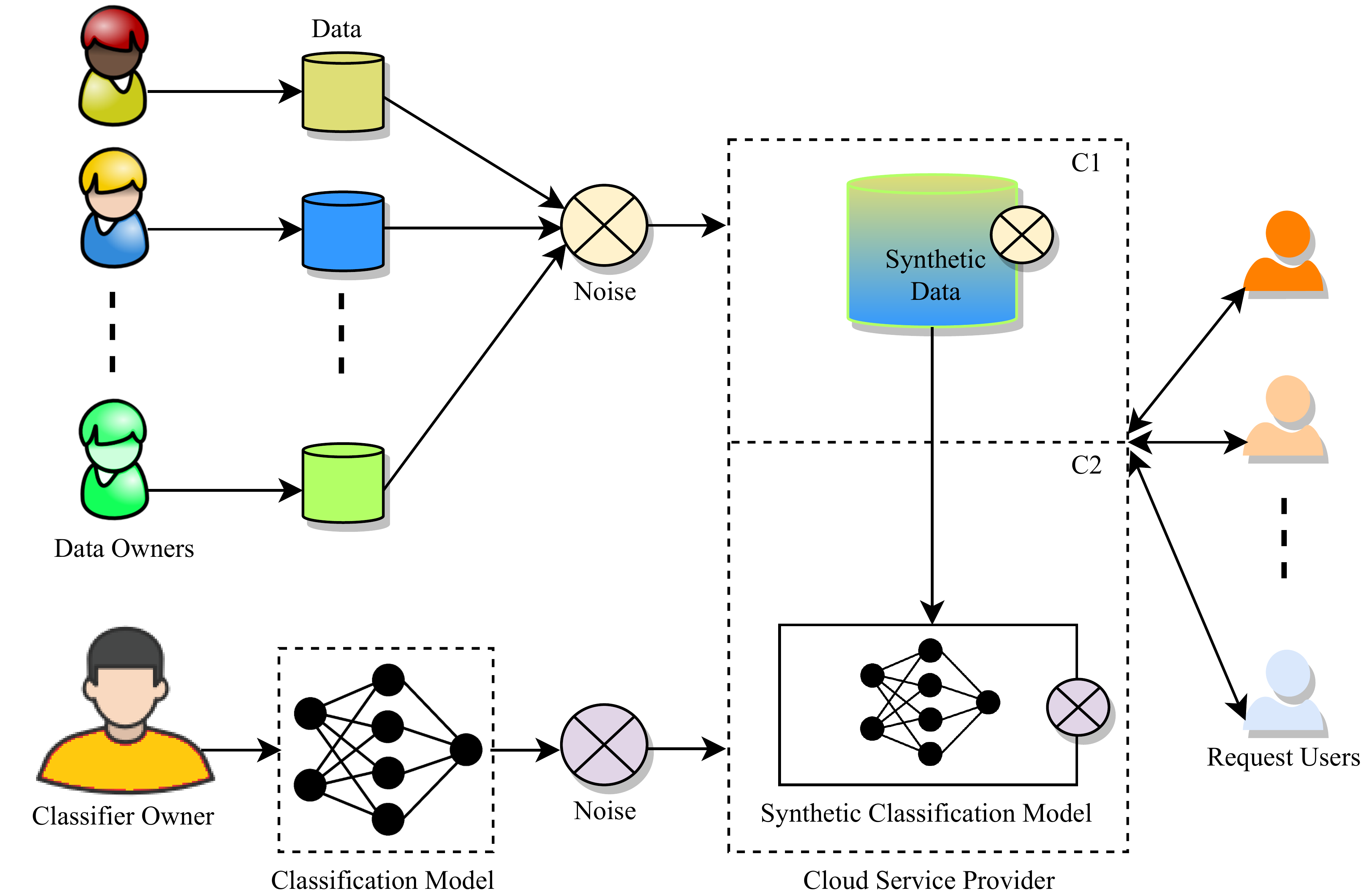}
		\caption{Bird eye view of DA-PMLM}
		\label{fig:23}
	\end{center}
\end{figure*}
\begin{itemize}
\item[$\bullet$] DA-PMLM permits several data owners and the classifier owner to share their data and machine learning model without hindrance. The Laplace mechanism has been used to keep the data and model private before sharing.
\item[$\bullet$] DA-PMLM uses two clouds: cloud1 and cloud2. Cloud1 handles the data storage and sharing process, whereas cloud2 carries out the classification tasks over accumulated data from various owners. To secure data and model with improving privacy, all entities are deemed untrustworthy.
\item[$\bullet$] DA-PMLM ensures the level of security because the value of the privacy budget is assigned according to the requirement of protection. 
\item[$\bullet$] Implementation and evaluation are performed using various datasets, which shows that DA-PMLM outperforms the state-of-art works in terms of privacy, accuracy, precision, recall, and f1-score.
\end{itemize} 
\textit{Organization:} Section 2 describes the related work followed by the proposed model and the process of preserving data in Sections 3 and 4, respectively. The privacy-preserving of the model and the data classification steps are presented in Section 5 and Section 6, respectively. The illustration of the proposed model is defined in Section 7. The operational design and computational complexity are discussed in Section 8. The experimental results with statistical analysis are shown in Section 9, followed by the conclusion and future work in Section 10. Table 1 demonstrates the list of symbols with their descriptive terms that have been used throughout this manuscript.
\begin{table}[h]
		\caption{List of Terminologies with their Explanatory Terms}
		\label{TableExp97}
          \centering
			\begin{tabular}{l l|l l|l l}\hline 
			\hline 
\textit{$CO$}: & Classifier Owner & \textit{$n^{**}$}: & Training objects
& \textit{$P$}: & Precision \\  
	
\textit{UE}: & Untrusted Entity &
\textit{$D_{i}$}: & Actual data & \textit{$N_{i}$}: & Noise \\ 

\textit{$CM$}: & Classification Model & \textit{$RU_{id}$}: & Request Users & \textit{$D_{i}^{N}$}: & Synthetic Data \\ 

\textit{$CA$}: & Classification Accuracy & \textit{$\Check{D}_{t^{'}}^{N}$}: & Training Data &
\textit{$R$}: & Recall \\

\textit{$\hat{D}_{i}^{N}$}: & Preprocessed Data & \textit{$\mathbb{C}^{L}$}: & Label Vector & \textit{$\Check{D}_{t^{''}}^{N}$}: & Testing Data \\

\textit{$\Delta$ f}: & Absolute distance & \textit{$n$}: & No. of data object & 
\textit{$FP$}: & False Positive \\  

\textit{DP}: & Differential Privacy  & \textit{$CS$}: & Cloud Storage & \textit{Pr}: & Probability\\

\textit{$\vartheta$}: & Scaling parameter & \textit{$n^{***}$}: & Testing objects & $\epsilon$: & Privacy budget \\

\textit{$\hat{R}$}: & Random function & \textit{f}: & Query function &  \textit{$n^{*}$}: & count of classes\\

\textit{$CM^{N}$}: & Synthetic Model & \textit{$CF$}: & Classifier & \textit{$FS$}: & F1-Score \\

\textit{$CSP$}: & Cloud Service Provider & \textit{$C1$}: & Cloud1 & \textit{$C2$}: & Cloud2 \\
\textit{$\sigma$}: & Standard deviation & \textit{$\mu$}: & Mean & \textit{$\Lambda$}: & Data attribute \\
 
				 \hline \hline
			\end{tabular}
	\end{table}
 \section{Related Work}
To protect the data, and machine learning model using the differential privacy mechanism has been categorized at the two levels: (a) At the data level, (b) At the model level.
 \subsection{At Data Level}
 To protect data privacy, Li et al. \cite{li2018privacy} proposed a differentially private scheme called Privacy-preserving Machine Learning with Multiple data providers (PMLM) with improved computational efficiency and data analysis accuracy. The PMLM scheme uses public-key encryption with a double decryption algorithm (DD-PKE) to transform the encrypted data into a randomized dataset and $\epsilon$-differential privacy to make the data private. But the PMLM scheme suffered from less accuracy as well as limited data sharing.
 Wang et al. \cite{wang2018privacy} proposed a distributed agent-based privacy-preserving framework, namely DADP, for real-time spatial statistics data collection and publication with an untrusted server. A distributed budget allocation mechanism and an agent-based dynamic grouping mechanism were developed to achieve global w-event $\epsilon$-differential privacy in a distributed manner. In DADP, crowd-sourced data is aggregated, and then the noise is added to it using the Laplace mechanism. However, it was considered a semi-centralized setting and resulted in a more complex system because it initiated a batch of trusted proxies (Agents) and anonymous connection technology to protect the privacy of users under an untrusted server.
A privacy-preserving deep learning model, namely PDLM, is presented in \cite{ma2018pdlm} that protects training data by encrypting it with multiple keys using a public-key distributed two trapdoors (DT-PKC) cryptosystem. The proposed scheme minimized the storage cost. However, it necessitates more calculations on the ciphertext, resulting in decreased efficiency.
 Hassan et al. \cite{hassan2019efficient} proposed an efficient privacy-preserving scheme based on machine learning to protect data privacy. The partially homomorphic encryption technique was utilized to re-encrypt data, and a differential privacy mechanism was used to add noise to the data. It permits all participants to publicly verify the validity of the ciphertext using a unidirectional proxy re-encryption (UPRE) scheme that reduces computational costs. But the proposed scheme has limited data sharing.
 A local differential privacy (LDP)-based classification algorithm for data centers is introduced \cite {fan2020privacy} with high efficiency and feasible accuracy. Sensitive information is protected by adding the noise using the Laplace mechanism in the pattern mining process. They devised a method for determining the level of privacy protection. However, the proposed method does not allow for data sharing.
To prevent information leakage, a framework, namely, noising before model aggregation federated learning (NbAFL), was proposed by Wei et al. \cite{wei2020federated}. Before the model aggregation, a differential privacy mechanism was utilized to add the noise to clients' local parameters. Nevertheless, this framework requires a significant amount of noise to add and sacrifices model utility.
 An adaptive privacy-preserving federated learning framework is introduced in \cite{liu2020adaptive}, which protects gradients by inserting the noise with varying privacy budgets. They devised a randomized privacy-preserving adjustment technology (RPAT) to enhance the accuracy, but the proposed framework makes the classification process computationally costly.
 A machine learning and probabilistic analysis-based model, namely MLPAM, is presented in \cite{gupta2020mlpam}. It supports various participants to share their data safely for different purposes using encryption, machine learning, and probabilistic approaches. This scheme reduced the risk associated with the leakage for prevention coupled with detection, but it does not provide privacy for the classifier. 
 To protect the Spatio-temporal aggregated data in real-time, Xiong et al. \cite{xiong2020real} proposed a privacy-preserving framework based on the local differential approach. The authors also devised a generalized randomized response (GRR) framework for obtaining reliable aggregated statistics continuously while maintaining user’s privacy. The proposed framework improves data usage but exposes the user's privacy by adding independent identically distributed noise to the associated data.
 Liu et al. \cite{liu2020local} proposed a differential privacy for local uncertain social network (DP-LUSN) model that protects the social network's community structured data. DP-LUSN increased the data utilization by minimizing the noise effect on a single edge. However, DP-LUSN only operates on static social networks.
 Sharma et al. \cite{sharma2021differential} proposed a Differential Privacy Fuzzy Convolution Neural Network framework, namely DP-FCNN, which protects the data and query processing by adding the noise using the Laplace mechanism.  The authors used the lightweight Piccolo algorithm to encrypt the data and the BLAKE2s algorithm to extract the key attributes from the data. However, DP-FCNN increases the computational overhead.
 
\subsection{At Model Level}
A privacy-preserving Naive Bayes learning scheme with various data sources using $\epsilon$-differential privacy and homomorphic encryption is presented in \cite{li2018differentially}. This scheme enables the trainer to train the Naive Bayes classier over the dataset provided jointly by the different data owners. But still, adversaries can forge and manipulate the data in this scheme.
Gao et al. \cite{GAO201872} proposed a privacy-preserving Naive Bayes classifier scheme to avoid information leakage under the substitution-then-comparison (STC) attack. They adopted a double-blinding technique to preserve the privacy of Naive Bayes. The proposed scheme reduced both the communication and computation overhead. But this scheme is unable to obtain the discovery of truth that protects privacy.
A privacy-preserving outsourced classification in the cloud computing (POCC) framework was introduced in \cite{li2018privacy1234}. It protects the privacy of sensitive data under various public keys using a fully homomorphic encryption proxy technique without leakage. However, encryption for outsourced data can protect privacy against unauthorized behaviors. It also makes effective data utilization, such as search over encrypted data, a complicated issue.
Liu et al. \cite{LIU2018807} proposed a private decision tree algorithm based on the noisy maximal vote. An effective privacy budget allocation strategy was adopted to make the balance between the true counts and noise.  It was constructed as an ensemble model with differential privacy to boost the accuracy and improve the stableness. In the proposed algorithm, the privacy analysis was performed only on each separate tree and not on the ensemble as a whole.
A differentially private gaussian processes classification (GPC) model was presented in \cite{xiong2019differential}, which adds noise to the classifier to provide privacy. The authors used the Laplacian approximation method to determine GPC's sensitivity and scaled noise to outputs generated from other dataset sections. In this manner, the total noise can be decreased. However, the proposed model is not sparse.
Wang and Zhang \cite{wang2020differential} proposed a differential privacy version of convex and nonconvex sparse classification approach based on alternating direction method of multiplier (ADMM) algorithm with mild conditions on the regularizers. By adding exponential noise to stable stages, they were able to turn the sparse problem into a multistep iteration process and accomplish privacy protection. However, estimating the noise introduced to stochastic algorithms in the proposed approach is problematic.
A differentially private ensemble learning method for classification, referred to as DPEL, was presented in \cite{LI202134} which provides privacy protection while ensuring prediction accuracy. The authors applied the Bag of Little Bootstrap and the Jaccard similarity coefficient techniques to enhance the ensemble's diversity. A privacy budget allocation strategy was designed to make the differentially private base classifiers by adding noise. But the DPEL method is affected by limited data sharing. Table 2 provides an overview of the literature review.
{\setlength\defaultaddspace{0.5ex}
\captionof{table}{Tabular sketch of the literature review}
\small
\begin{tabularx}{\linewidth}{|p{1.8cm}|p{3.5cm}|p{3cm}|p{1.7cm}|}

\hline
        Model/Scheme /Framework & Workflow \& Implementation & Outcomes 
        & Drawback 
        \\
        \hline  \hline
 A machine learning-based privacy-preserving model for data protection \cite{li2018privacy}
        & 
    \vspace{-0.3cm} 
        \begin{itemize}[topsep=-0.3cm,leftmargin=0.2cm,label=\textbullet]
  \item To prevent data leakage, additively homomorphic encryption and differential privacy techniques were employed
     \item Abalone, Wine, Cpu, Glass, and Krkopt data sets were used for experiments 
        \end{itemize} & 
      \vspace{-0.3cm} 
        \begin{itemize}[topsep=-0.3cm,leftmargin=0.2cm,label=\textbullet]
            \item Improve the accuracy of data analysis and the efficiency of computations
            \item The model seems to be more protected as per the security analysis
        \end{itemize} & 
       For the classifier, it does not provide privacy \\
        \hline   
 A distributed privacy-preserving task allocation framework \cite{wang2018privacy}    & 
    \vspace{-0.3cm} 
        \begin{itemize}[topsep=-0.3cm,leftmargin=0.2cm,label=\textbullet]
            \item $w$ event and $\epsilon$-Differential Privacy mechanisms were used to protect data
            \item Python language and taxi trajectory, nice rider datasets were used for experiments 
        \end{itemize} & 
      \vspace{-0.3cm} 
        \begin{itemize}[topsep=-0.3cm,leftmargin=0.2cm,label=\textbullet]
            \item Less Mean Absolute Error (MAE) and Mean Relative Error (MRE)
\item Enhance the overall throughput of the whole vehicular network  
        \end{itemize} & 
A trade-off between accuracy and the overall payment of crowd-sensing server
  \\
        \hline   
A cloud-based Deep Learning model for data privacy  \cite{ma2018pdlm} & 
    \vspace{-0.3cm} 
        \begin{itemize}[topsep=-0.3cm,leftmargin=0.2cm,label=\textbullet]
            \item A homomorphic encryption mechanism was adopted to encrypt the data
            \item To assess PDLM performance, a LeNet deep learning model was utilized
        \end{itemize} & 
      \vspace{-0.3cm} 
        \begin{itemize}[topsep=-0.3cm,leftmargin=0.2cm,label=\textbullet]
            \item Achieve high accuracy up to 94\%
            \item Perform the tasks of privacy-preserving training and classification appropriately
        \end{itemize} & The use of multiple keys for big data reduced the efficiency of the machine learners

  \\
        \hline  
A privacy-preserving with public verifiability scheme \cite{hassan2019efficient}   & 
    \vspace{-0.3cm} 
        \begin{itemize}[topsep=-0.3cm,leftmargin=0.2cm,label=\textbullet]
            \item Data is protected by applying differential privacy and partially homomorphic encryption techniques
            \item Java pairing based cryptography (JPBC) library was used for experiments
        \end{itemize} & 
      \vspace{-0.3cm} 
        \begin{itemize}[topsep=-0.3cm,leftmargin=0.2cm,label=\textbullet]
            \item Lower computational as well as communication costs
            \item The security analysis verifies that the scheme is more secure under the random oracle model
        \end{itemize} & With homomorphic encryption, the issue of ciphertext expansion arise
 \\
        \hline   
 A classification algorithm satisfies local differential privacy \cite {fan2020privacy} & 
    \vspace{-0.3cm} 
        \begin{itemize}[topsep=-0.3cm,leftmargin=0.2cm,label=\textbullet]
            \item Internet of Multimedia Things (IoMT) produces the sensitive information
            \item Laplace noise is injected into the data during pattern mining
        \end{itemize} & 
      \vspace{-0.3cm} 
        \begin{itemize}[topsep=-0.3cm,leftmargin=0.2cm,label=\textbullet]
            \item High efficiency, reliability, and precision
            \item Less MAE, mean squared error (MSE) metrics
        \end{itemize} & 
IoMT devices need higher latency and bandwidth \\
        \hline   
A framework protects training data based on differential privacy  \cite{wei2020federated}
        & 
    \vspace{-0.3cm} 
        \begin{itemize}[topsep=-0.3cm,leftmargin=0.2cm,label=\textbullet]
            \item Inject the noise into the parameters of the local learning model
            \item This framework was assessed by utilizing multi-layer perception and MNIST datasets
        \end{itemize} & 
      \vspace{-0.3cm} 
        \begin{itemize}[topsep=-0.3cm,leftmargin=0.2cm,label=\textbullet]
            \item A trade-off between federated learning performance and privacy
            \item  Acquire high privacy using random mini-batches
        \end{itemize} & 
        Less accuracy due to injected noise in the parameters \\
        \hline  
A secure multi-party computation strategy by adaptive privacy-preserving \cite{liu2020adaptive} & 
    \vspace{-0.3cm} 
        \begin{itemize}[topsep=-0.3cm,leftmargin=0.2cm,label=\textbullet]
            \item The Laplace mechanism was adopted to add noise with different
privacy budget
            \item TensorFlow library and MNIST dataset were used for experiments
        \end{itemize} & 
      \vspace{-0.3cm} 
        \begin{itemize}[topsep=-0.3cm,leftmargin=0.2cm,label=\textbullet]
            \item Achieve high accuracy up to 88.46\%
            \item Less computation and transmission overhead
        \end{itemize} & 
 Data perturbation may degrade the data utility\\
        \hline   
A secure data sharing model for determining guilty entities against data leakage \cite{gupta2020mlpam}
        & 
    \vspace{-0.3cm} 
        \begin{itemize}[topsep=-0.3cm,leftmargin=0.2cm,label=\textbullet]
            \item The attribute based encryption and differential privacy were used to preserve the privacy of multiple participants' data
            \item The pyhton language and Glass, Iris, Wine, and Balance Scale datasets were utilized for experiments
        \end{itemize} & 
      \vspace{-0.3cm} 
        \begin{itemize}[topsep=-0.3cm,leftmargin=0.2cm,label=\textbullet]
            \item Perform secure data distribution among users by an effective allocation mechanism
            \item Achieve high accuracy up to 97\%
        \end{itemize} & 
      High computational and communication costs while transferring the data \\
        \hline
  A real-time spatio-temporal data aggregation method \cite{xiong2020real}
           & 
   \vspace{-0.3cm} 
        \begin{itemize}[topsep=-0.3cm,leftmargin=0.2cm,label=\textbullet]
            \item LDP was employed for data protection
            \item MATLAB R2018a programming language and Geolife, Taxi Service datasets were used to conduct the experimentations
        \end{itemize} & 
      \vspace{-0.3cm} 
        \begin{itemize}[topsep=-0.3cm,leftmargin=0.2cm,label=\textbullet]
            \item Less MAE, Mean MAE (MMAE), Mean KL divergence (MKLD), and Mean RelativeError (MRE)
            \item Acquire adequate trade-off of privacy and utility 
        \end{itemize} & 
 High computation and communication costs due to generalized randomized response\\
        \hline   
A local differential privacy scheme for social network publishing     \cite{liu2020local}       & 
   \vspace{-0.3cm} 
        \begin{itemize}[topsep=-0.3cm,leftmargin=0.2cm,label=\textbullet]
            \item The perturbed local method was applied to yield a synthetic network
            \item The experiments are performed using Python language, and WebKB, Citation, and Cora datasets
        \end{itemize} & 
      \vspace{-0.3cm} 
        \begin{itemize}[topsep=-0.3cm,leftmargin=0.2cm,label=\textbullet]
            \item Decrease network structure information's loss
            \item This model provides a substantial degree of privacy
        \end{itemize} & 
 It does not support dynamic social networks\\
        \hline   
A fuzzy convolution neural network framework based on Differential Privacy for data security   \cite{sharma2021differential}     & 
   \vspace{-0.3cm} 
        \begin{itemize}[topsep=-0.3cm,leftmargin=0.2cm,label=\textbullet]
            \item Data is encrypted with the help of a 128-bit block cipher
            \item Java Development toolkit (JDK) 1.8, Weka tools, and Adult and Heart disease dataset were used for experimentations
        \end{itemize} & 
      \vspace{-0.3cm} 
        \begin{itemize}[topsep=-0.3cm,leftmargin=0.2cm,label=\textbullet]
            \item The comparative analysis shows that it takes less run-time up to 14 ms
            \item More effective in terms of scalability and accuracy
        \end{itemize} & 
 Not suitable for hybrid machine learning algorithms\\
        \hline   
A privacy-preserving learning algorithm based on differential privacy\cite{li2018differentially}
        & 
      \vspace{-0.3cm} 
        \begin{itemize}[topsep=-0.3cm,leftmargin=0.2cm,label=\textbullet]
            \item The differential privacy and homomorphic encryption mechanisms were used for data preservation
            \item A GMP library was adopted to execute cryptographic operations
        \end{itemize} & 
      \vspace{-0.3cm} 
        \begin{itemize}[topsep=-0.3cm,leftmargin=0.2cm,label=\textbullet]
            \item Achieve secure data sharing without breaking privacy
            \item Require less computational time for training model
        \end{itemize} & 
        Less efficiency because of encryption \\
        \hline 
A privacy-preserving scheme for protecting the trained model  \cite{GAO201872}
        & 
    \vspace{-0.3cm} 
        \begin{itemize}[topsep=-0.3cm,leftmargin=0.2cm,label=\textbullet]
            \item Additively homomorphic encryption was used to encrypt the data
            \item The experiments were conducted over Breast Cancer Wisconsin (Original), Statlog Heart datasets
        \end{itemize} & 
      \vspace{-0.3cm} 
        \begin{itemize}[topsep=-0.3cm,leftmargin=0.2cm,label=\textbullet]
            \item Less the computation overhead because of the offline stage of the server
            \item Accomplish suitable privacy-preserving training tasks
        \end{itemize} & 
      It does not provide efficient information sharing \\
        \hline

A framework for protecting the outsourced classification model
 \cite{li2018privacy1234}
        & 
    \vspace{-0.3cm} 
        \begin{itemize}[topsep=-0.3cm,leftmargin=0.2cm,label=\textbullet]
            \item Data was encrypted by using the multikey homomorphic encryption proxy
            \item The experiments were carried out using a naive bayes classifier
        \end{itemize} & 
      \vspace{-0.3cm} 
        \begin{itemize}[topsep=-0.3cm,leftmargin=0.2cm,label=\textbullet]
            \item More effective for data utilization
            \item Achieve high-level data protection
        \end{itemize} & 
    Less efficiency due to additions, multiplications, and the nonlinear computations over encrypted data  \\
        \hline
A method for learning a differentially private decision tree \cite{LIU2018807}    & 
   \vspace{-0.3cm} 
        \begin{itemize}[topsep=-0.3cm,leftmargin=0.2cm,label=\textbullet]
            \item The differentially-private decision tree was trained by leveraging the noisy maximal vote
            \item The experiments were carried out over Mushroom and Adult datasets
        \end{itemize} & 
      \vspace{-0.3cm} 
        \begin{itemize}[topsep=-0.3cm,leftmargin=0.2cm,label=\textbullet]
            \item Improve the learned classifier's utility
        \item Provide high accuracy  up to 81.19\%
        \end{itemize} & 
 The ensemble model's performance can not be maximized\\
        \hline   
A privacy-preserving model for data classification \cite{xiong2019differential}   & 
   \vspace{-0.3cm} 
        \begin{itemize}[topsep=-0.3cm,leftmargin=0.2cm,label=\textbullet]
            \item A variant-noise mechanism was used to perturb the classifier with the different scaled noise
            \item The experiments are conducted over an actual oil dataset provided by NCRG of Aston University
        \end{itemize} & 
      \vspace{-0.3cm} 
        \begin{itemize}[topsep=-0.3cm,leftmargin=0.2cm,label=\textbullet]
            \item Achieve elevated accuracy up to 75.00\%
            \item Retain the classifier's accuracy
        \end{itemize} & 
Noise inevitably decreases the accuracy \\
        \hline   
A privacy-preserving for convex and nonconvex sparse classification \cite{wang2020differential}  & 
   \vspace{-0.3cm} 
        \begin{itemize}[topsep=-0.3cm,leftmargin=0.2cm,label=\textbullet]
            \item The sparse problem was transformed into a multistep iteration process through ADMM algorithm
            \item The logistic regression and KDDCup99 dataset were used for experiments
        \end{itemize} & 
      \vspace{-0.3cm} 
        \begin{itemize}[topsep=-0.3cm,leftmargin=0.2cm,label=\textbullet]
            \item More effective and efficient for performing sensitive data analysis
            \item Accomplish appropriate privacy-preserving training and classification tasks
        \end{itemize} & 
It does not share the data effectively in several environments \\
        \hline   
A privacy-preserving method for preventing leakage in the  process of classification \cite{LI202134} & 
   \vspace{-0.3cm} 
        \begin{itemize}[topsep=-0.3cm,leftmargin=0.2cm,label=\textbullet]
            \item Laplace noise is injected into the classifier's objective function
            \item The experiments were run over Wisconsin, User knowledge, Mushroom, Adult, and Bank Market datasets

        \end{itemize} & 
      \vspace{-0.3cm} 
        \begin{itemize}[topsep=-0.3cm,leftmargin=0.2cm,label=\textbullet]
            \item Perform a trade-off between privacy and accuracy
            \item Acquire high accuracy up to 94.66\%
        \end{itemize} & 
 It does not provide the data protection \\
        \hline   
\end{tabularx}
\label{table:ressuffixes}
}
\par
In the light of the aforementioned discussion, it can be concluded that earlier works limited data sharing and injected the noise into either data or classification model and/or protected them by employing diverse encryption techniques followed by the machine learning-based classification, resulting in decreased privacy, accuracy, and higher computing cost. In addition, the prior models only supported a single owner or untrustworthy entity. Unlike the previous works, DA-PMLM ensures that both data and the classification model are protected before being shared. The proposed model enables numerous data and classifier owners to securely share outsourced data and classification model while assuming all entities to be untrustworthy.
\section{PROPOSED MODEL}
The proposed model (Fig. 2) consists of four entities: Data Owners (\textit{$DO_{id}$}), Classifier Owner (\textit{$CO$}), Cloud Service Provider (\textit{$CSP$}), and Request Users (\textit{$RU_{id}$}), which are described as follows in terms of intercommunication and vital information flow:
\begin{enumerate}
 \item[1)]\textit{$DO_{id}$}: An entity that yields the data and appeals to \textit{$CSP$} for storing, computing, and sharing services. To protect the data, \textit{$DO_{id}$} injects the noise into data using the $\epsilon$-differential privacy mechanism prior send it to \textit{$CSP$}. \textit{$DO_{id}$} permits \textit{$CSP$} to share the data with \textit{$RU_{id}$}. Since it is assumed that \textit{$DO_{id}$} can not misuse its data, it may disclose the data of other owners. Therefore, \textit{$DO_{id}$} is not considered a trusted entity.
 \item[2)] \textit{$CO$}: An entity that has a classification model (\textit{CM}) and offers the services to perform the classification tasks through \textit{$CSP$}. Before passing \textit{CM} to \textit{$CSP$}, \textit{$CO$} injects the noise utilizing the $\epsilon$-differential privacy mechanism to retain the privacy and stores it. \textit{$CO$} is viewed as an untrustworthy entity in the proposed model.
 \item[3)] \textit{$CSP$}: An entity that collects all the data from \textit{$DO_{id}$} and \textit{CM} from \textit{$CO$}. It provides storage, computation, and data sharing services to \textit{$DO_{id}$}, \textit{$CO$}, and \textit{$RU_{id}$}. It also offers the classification services using \textit{CM} to \textit{$DO_{id}$}, and \textit{$RU_{id}$}. \textit{$CSP$} trains \textit{CM} using the machine learning algorithm over collected data and obtains classified data from \textit{CM}. These accessed results are shared between  \textit{$DO_{id}$}, or \textit{$RU_{id}$}. \textit{$CSP$} is introduced as a semi-trusted entity in the proposed model, as it strictly follows the protocol but is curious to learn the information. In our model, the Cloud Platform (\textit{CP}) consists of two clouds: Cloud1 (\textit{C1}) and Cloud2 (\textit{C2}). The data is stored in Cloud Storage (\textit{CS}) at \textit{C1} while the classification model is kept in Classifier (\textit{CF}) at \textit{C2}. \textit{$CSP$} is the only entity that acts as a bridge among \textit{$DO_{id}$}, \textit{$CO$}, and \textit{$RU_{id}$}. It stores the data and model to perform the data sharing and classification tasks.
 \item[4)] \textit{$RU_{id}$}: An entity that requests \textit{$CSP$} to obtain the owner's data, and classification service. \textit{$RU_{id}$} receives the data from \textit{C1} and accesses the \textit{CM} from \textit{C2} through the \textit{$CSP$}. \textit{$RU_{id}$} is treated as an untrusted entity in our model. 
  \end{enumerate}
 	\begin{figure*}[!htbp]
	\begin{center}
	\includegraphics[width=0.85\textwidth]{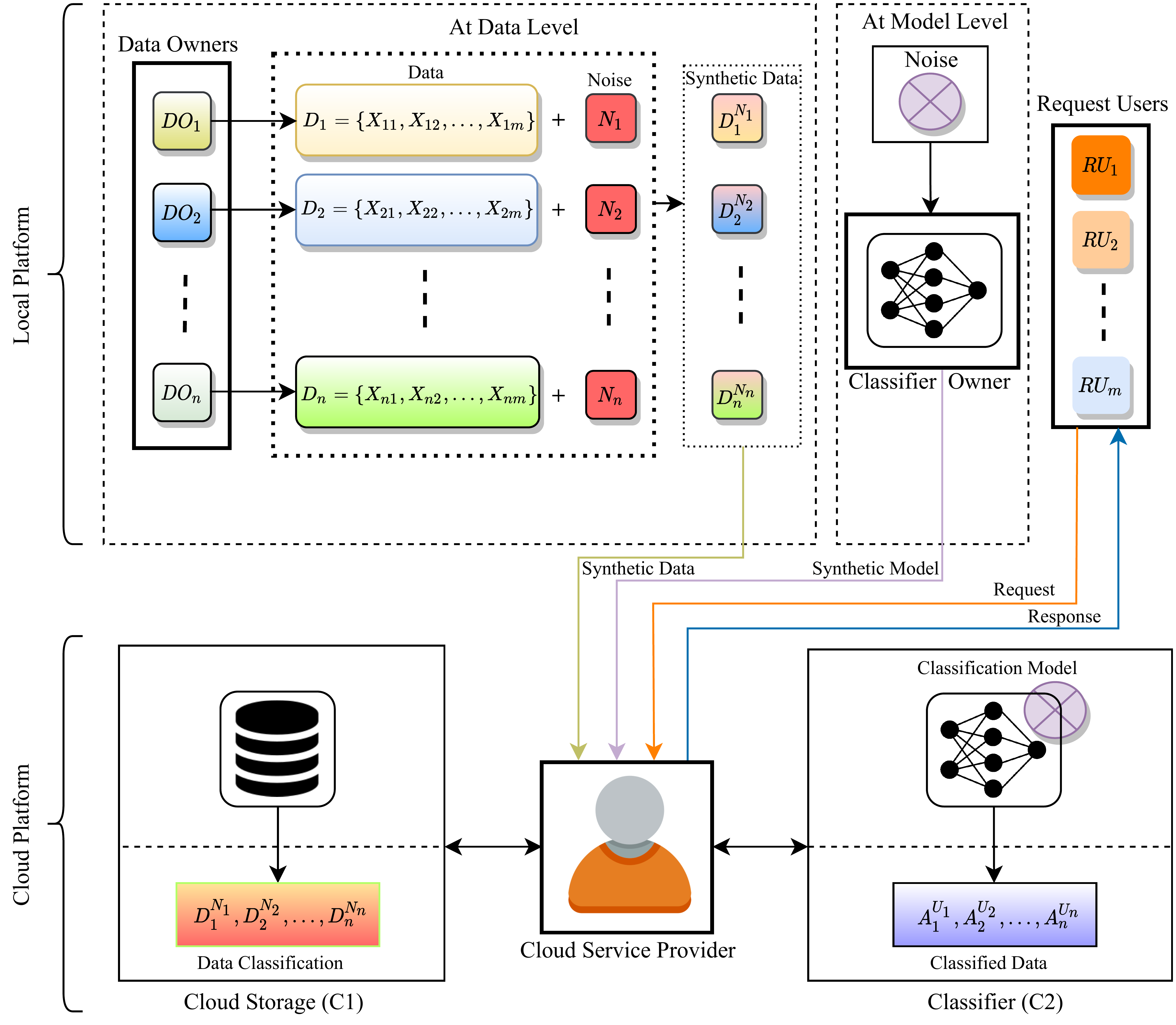}
		\caption{Proposed DA-PMLM architecture }
		\label{fig:CPABE24}
	\end{center}
\end{figure*}
 \par
 Consider the data owners $\mathbb{DO}$ = \{$DO_{1}$, $DO_{2}$, $\dots$, $DO_{n}$\} having data $\mathbb{D}$, such as $\mathbb{D}$ $\in$ \{$D_{1}$, $D_{2}$, $\dots$, $D_{n}$\}. It is essential that $\mathbb{DO}$ must share $\mathbb{D}$ among the other parties, such as \textit{$CSP$}, and several request users $\mathbb{RU}$ = \{$RU_{1}$, $RU_{2}$, $\dots$, $RU_{m}$\} for storage, computation, and data usage purposes. Due to the sensitive data, $\mathbb{DO}$ do not aspire to disclose $\mathbb{D}$ to unauthorized parties. 
 As a result, each $\mathbb{DO}$ procures synthetic data $\mathbb{D^{N}}$ = \{$D_{1}^{N_{1}}$, $D_{2}^{N_{2}}$, $\dots$, $D_{n}^{N_{n}}$\} by injecting noise $\mathbb{N}$ = \{$N_{1}$, $N_{2}$, $\dots$, $N_{n}$\} into sensitive data \{$D_{1}$, $D_{2}$, $\dots$, $D_{n}$\} using $\epsilon$-differential privacy to make the data private before sharing (details discussion in Section 4). DA-PMLM makes the use of Differential Privacy (\textit{DP}) mechanism to inject $N_{1}$, $N_{2}$, $\dots$, $N_{n}$ into $D_{1}$, $D_{2}$, $\dots$, $D_{n}$, respectively, because \textit{DP} is the most appropriate mechanism for generating $N_{1}$, $N_{2}$, $\dots$, $N_{n}$ and sharing data $D_{1}^{N_{1}}$, $D_{2}^{N_{2}}$, $\dots$, $D_{n}^{N_{n}}$ without sacrificing or revealing sensitive and personal information about individual owners \cite{chen2014correlated}. Each $DO_{1}$, $DO_{2}$, $\dots$, $DO_{n}$ delivers $D_{1}^{N_{1}}$, $D_{2}^{N_{2}}$, $\dots$, $D_{n}^{N_{n}}$ to \textit{$CSP$} that store it on \textit{CS} at \textit{C1} for computation, and sharing purposes. Due to the lack of storage space and the need to handle multiple requests, \textit{$CO$} aspires to store \textit{CM} on the cloud. Nevertheless, \textit{$CO$} injects noise $\mathbb{\hat{N}}$ into \textit{CM} to protect privacy before transferring it (as described in Section 5). \textit{$CO$} sends the synthetic model ($\textit{$CM^{N}$}$) to \textit{$CSP$} that keeps it on \textit{CF} at \textit{C2}. \textit{$CSP$} utilizes the stored data $D_{1}^{N_{1}}$, $D_{2}^{N_{2}}$, $\dots$, $D_{n}^{N_{n}}$ to accomplish the classification tasks to make a fit $\textit{$CM^{N}$}$, and also share it with $RU_{1}$, $RU_{2}$, $\dots$, $RU_{m}$. 
\textit{$CSP$} also uses the stored model $\textit{$CM^{N}$}$ to handle the query of $DO_{id}$, and $RU_{id}$. $DO_{1}$, $DO_{2}$, $\dots$, $DO_{n}$ and  $RU_{1}$, $RU_{2}$, $\dots$, $RU_{m}$ can make any query and send it to \textit{$CSP$}. 
 Afterward, \textit{$CSP$} interacts with \textit{CF} and achieves the results of these queries from $\textit{$CM^{N}$}$. The obtained results are sent to the corresponding entity $DO_{1}$, $DO_{2}$, $\dots$, $DO_{n}$, and  $RU_{1}$, $RU_{2}$, $\dots$, $RU_{m}$ by \textit{$CSP$}. 
\section{Privacy-Preserving of Data}
Data owners ($DO_{1}$, $DO_{2}$, $\dots$, $DO_{n}$) injects the noise ($N_{1}$, $N_{2}$, $\dots$, $N_{n}$) into data ($D_{1}$, $D_{2}$, $\dots$, $D_{n}$), respectively, using $\epsilon$-differential privacy. A random function $\hat{R}$ is referred to as $\epsilon$-differential privacy \cite{mcsherry2009privacy} if for any combination of dataset $D$ and its neighbor $D^{'}$, and for all $\vartheta$ $\subseteq$ Range($\hat{R}$) are defined using Eq. (1):  
\begin{equation}
Pr[\hat{R}(D) = \vartheta] \leq exp(\epsilon) \times Pr[\hat{R}(D^{'}) = \vartheta]
\end{equation}
where $Pr[.]$ is the probability function that has been applied to the function $\hat{R}$ and demonstrates the privacy disclosure risk. The privacy budget ($\epsilon$) is defined by the results of statistical computation, analysis, and the individual's information that has to be kept private. To achieve a higher level of privacy protection, the value of $\epsilon$ should be reduced. A numeric query function ($f$) that is utilized to map the data set $D$ into real space with $d$-dimensional, assigned as $f$: $D$ $\rightarrow$ $R^{d}$. The sensitivity of $f$ for all the combinations of dataset $D$ and $D^{'}$ is defined using Eq. (2):
\begin{equation}
\Delta f = \max\limits_{D, D^{'}} \parallel f(D) - f (D^{'}) \parallel_{P_{1}}
\end{equation}
where $\parallel$ . $\parallel_{P_{1}}$ is norm. The sensitivity $\Delta f$ is only depend on the query function $f$ and finds the maximum gap of query results on neighboring datasets. For any function  $f$: $D$ $\rightarrow$ $R^{d}$, the Laplace mechanism $F$ is defined using Eq. (3). The mechanism $F$ takes input $\chi$, and $\epsilon$ $>$ 0, a numeric query function $f$, and computes the output. It is based on the sensitivity of $f$ and receives the statistical noise from a Laplace distribution. Afterward, the obtained noise is injected into the datasets.
\begin{equation}
F(\chi) = f(\chi) + ( Lap_{1}(\varpi), Lap_{2}(\varpi), \dots, Lap_{d}(\varpi))
\end{equation}
where the noise $Lap_{j}$($\varpi$) ($j$ $\in$ [$1$, $d$]), and $\varpi$ belongs to $R^{+}$ comes from a Laplace distribution, whose the probability density function is calculated using Eq. (4). 
\begin{equation}
    N = \frac{1}{2\varpi} \cdot (exp(\frac{-|\chi|}{\varpi}))
\end{equation}
where $N$ is a noise vector. The noise $N_{1}$, $N_{2}$, $\dots$, $N_{n}$ is generated taking the sample from the Laplace distribution with scaling parameter $\varpi$ = $\Delta f/\epsilon$. The parameter $\varpi$ is under the control of $\epsilon$. The generated noise vector $N_{1}$, $N_{2}$, $\dots$, $N_{n}$ is injected into the corresponding $D_{1}$, $D_{2}$, $\dots$, $D_{n}$ as $D_{i}^{N_{i}}$ = $D_{i}$ + $N_{i}$, where $i$ $\in$ [$1$, $n$]. After adding $N_{1}$, $N_{2}$, $\dots$, $N_{n}$, $DO_{1}$, $DO_{2}$, $\dots$, $DO_{n}$ obtains the Synthetic Data $D_{1}^{N_{1}}$, $D_{2}^{N_{2}}$, $\dots$, $D_{n}^{N_{n}}$. $DO_{1}$, $DO_{2}$, $\dots$, $DO_{n}$ sends data $D_{1}^{N_{1}}$, $D_{2}^{N_{2}}$, $\dots$, $D_{n}^{N_{n}}$ to \textit{$CSP$} that stores it on \textit{CS} at \textit{C1}, shares it among $RU_{1}$, $RU_{2}$, $\dots$, $RU_{m}$ and performs the classification task over it.   
\section{Privacy-Preserving of Classification Model}
To protect the privacy of \textit{$CM$}, \textit{$CO$} injects the noise $\mathbb{\hat{N}}$ into \textit{$CM$} before sharing it on the cloud. In DA-PMLM, \textit{$CO$} has a Naive Bayes (\textit{NB}) classifier, which is used to perform the classification tasks.    
The \textit{NB} classifier takes the input vector $\mathbb{D^{N}}$ = \{$x_{1}$, $x_{2}$, $\dots$, $x_{d}$\}, which belongs to $R^{d}$, and classify $d$ into a class in the form of discrete set $\mathbb{C}$ = \{$C_{1}$, $C_{2}$, $\dots$, $C_{n}$\}. The functionality of the \textit{NB} classifier based on the Bayes decision rule \cite{berrar2018bayes} is to select the class $\mathbb{C}$ with the highest posterior probability of $\mathbb{D^{N}}$ and measures it by using Eq. (5). 
\begin{equation}
 \mathbb{C^{L}} =  \argmax_{i \in n} p(\mathbb{C} = C_{i}|\mathbb{D^{N}}=x)
\end{equation}
The highest posterior probability's value ($p(\mathbb{C} = C_{i}, \mathbb{D^{N}}=x)$) is equal to the factorization of $p(\mathbb{C}$ = $C_{i})$ $\prod_{j=1}^{d}$ $P(\mathbb{D}^{N}=x_{j}|\mathbb{C}=C_{i})$. Each component of $x$ is conditionally independent based on the attribute conditional assumption. It can also be expressed as ($log$ $p(\mathbb{C} = C_{i})$ = $\sum_{j=1}^{d}$ $log$ $p(\mathbb{D}^{N}=x_{j}$)) in a logarithm form. The conditional probabilities $P(\mathbb{D}^{N}=x_{j}|\mathbb{C}=C_{i})$ is computed using the mean ($\mu_{i,j}$), and variance ($\sigma_{i,j}^{2}$), which are calculated for class $j$ using the values of attribute $X$ from the training set. As a result, the sensitivity of both $\mu_{i}$, and standard deviation ($\sigma_{i,j}$) are computed. It is assumed that the feature value of $\mathbb{D}_{j}^{N}$ is bounded in range [$g_{i}$, $h_{i}$]. The sensitivity of $\mu_{i,j}$ is calculated using ($h_{i}$ - $g_{i}$)/($n$+1), and the sensitivity of $\sigma_{i,j}$ is computed using $n$ $*$ ($h_{i}$ - $g_{i}$) / ($n$+1). After calculating the $\mu_{i,j}$ and $\sigma_{i,j}$, \textit{$CO$} generates the noise using the Eqs.1 to 4. The obtained noise $\hat{N}_{i}$ and $\hat{\hat{N}}_{i}$ are injected into $\mu_{i}$ and $\sigma_{i}$, respectively, using Eqs. 6 and 7.   
\begin{equation}
   \mu_{i}^{'} =  \mu_{i} + \hat{N}_{i}
\end{equation}
\begin{equation}
   \sigma_{i}^{'} = \sigma_{i} + \hat{\hat{N}}_{i}
\end{equation}
\textit{$CO$} uses $\mu_{i}^{'}$, and $\sigma_{i}^{'}$ to calculate $p(\mathbb{C} = C_{i}, \mathbb{D^{N}}=x)$. Afterword, synthetic model $\textit{$CM^{N}$}$ is sent to \textit{$CSP$} to handle the queries of $DO_{1}$, $DO_{2}$, $\dots$, $DO_{n}$, and $RU_{1}$, $RU_{2}$, $\dots$, $RU_{m}$. Therefore, unauthorized users or any attacker are  unable to acquire the sensitive classification model from \textit{$CSP$}; the proof is provided in the following theorem:        
 \begin{theorem}
In the proposed model, the privacy-preserving mechanism of the classification model satisfies the parallel composition of $\epsilon$-differential privacy.
	\end{theorem}
\begin{proof}
Let $\hat{R}_{1}$, $\hat{R}_{2}$, $\dots$, $\hat{R}_{n}$ be $n$ mechanisms, where each mechanism $\hat{R}_{i}$ ($i$ $\in$ [1, $n$]) provides $\epsilon$-differential privacy. Let $\mu_{1}$, $\mu_{2}$, $\dots$, $\mu_{n}$ and $\sigma_{1}$, $\sigma_{2}$, $\dots$, $\sigma_{n}$ are $n$ arbitrary mean and standard deviation to calculate the conditional probability attributes. For a new mechanism $\hat{R}$, the sequence of $\hat{R}$($\hat{R}_{1}$($\mu_{1}$), $\hat{R}_{2}$($\mu_{2}$), $\dots$, $\hat{R}_{n}$($\mu_{n}$)) and $\hat{R}$($\hat{R}_{1}$($\sigma_{1}$), $\hat{R}_{2}$($\sigma_{2}$), $\dots$, $\hat{R}_{n}$($\sigma_{n}$)) provide $max_{1 \leq i \leq n}$ $\epsilon$- differential privacy.
$CO$ generates the Laplace noise $\hat{N}_{i}$ and $\hat{\hat{N}}_{i}$ according to sensitivity ($\Delta f$), privacy budget ($\epsilon$) to inject the produced noise into $\mu_{1}$, $\mu_{2}$, $\dots$, $\mu_{n}$ and $\sigma_{1}$, $\sigma_{2}$, $\dots$, $\sigma_{n}$, respectively, which satisfy
 $Pr$[$\hat{R}(\mu)$ = $\vartheta$] $\leq$ $exp$($\epsilon$) $\times$ $Pr$[$\hat{R}(\mu^{'})$ = $\vartheta$], and $Pr$[$\hat{R}(\sigma)$ = $\vartheta$] $\leq$ $exp$($\epsilon$) $\times$ $Pr$[$\hat{R}(\sigma^{'})$ = $\vartheta$].
 It can be observed that the mechanism $\hat{R}$ provides privacy to the classification model. Hence this privacy-preserving mechanism of the classification model satisfies the parallel composition of $\epsilon$-differential privacy.
\end{proof}
\section{Data Classification}
\textit{$CSP$} receives the synthetic data $D_{1}^{N_{1}}$, $D_{2}^{N_{2}}$, $\dots$, $D_{n}^{N_{n}}$ from $DO_{1}$, $DO_{2}$, $\dots$, $DO_{n}$ and prepossesses it by applying the normalization function shown in Eq. (8), where $D_{i}^{N_{i}}$ is training sample, $\mu$, and $\sigma$ are the mean and the standard deviation of the training sample, respectively. 
\begin{equation}
    \Check{D}_{i}^{N_{i}} = \frac{(D_{i}^{N_{i}} - \mu)}{\sigma}
\end{equation}
\textit{$CSP$} achieves the prepossessed data $\Check{\mathbb{D}}_{i}^{N_{i}}$ = \{$\Check{D}_{1}^{N_{1}}$, $\Check{D}_{2}^{N_{2}}$, $\dots$, $\Check{D}_{n}^{N_{n}}$\}, which belongs to $n^{*}$ is less than equal to $n$ classes $\mathbb{C}$ = \{$C_{1}$, $C_{2}$, $\dots$, $C_{n^{*}}$\}, for all $\mathbb{C}_{i}$ = $\mathbb{D}$, $i$ $\in$ [$1$, $n^{*}$] and $\mathbb{C}_{i}$ $\cap$ $\mathbb{C}_{k}$ = $\phi$, where $k$ $\in$ [$1$, $n^{*}$] $\wedge$  $i \neq k$. Fig. 3 depicts the stepwise process for data classification in which the prepossessed data $\Check{\mathbb{D}}_{i}^{N_{i}}$ = \{$\Check{D}_{1}^{N_{1}}$, $\Check{D}_{2}^{N_{2}}$, $\dots$, $\Check{D}_{n}^{N_{n}}$\} is split into training data $\Check{\mathbb{D}}_{\grave{t}}^{N}$ = \{$\Check{D}_{\grave{t},1}^{N_{1}}$, $\Check{D}_{\grave{t},2}^{N_{2}}$, $\dots$, $\Check{D}_{\grave{t},n^{**}}^{N_{n^{**}}}$\}, and testing data $\Check{\mathbb{D}}_{\Grave{\Grave{t}}}^{N}$ = \{$\Check{D}_{\Grave{\Grave{t}},1}^{N_{1}}$, $\Check{D}_{\Grave{\Grave{t}},2}^{N_{2}}$, $\dots$, $\Check{D}_{\Grave{\Grave{t}},n^{***}}^{N_{n^{***}}}$\} fulfill the following conditions a) $\Check{\mathbb{D}}_{\grave{t}}^{N}$ $\cup$ $\Check{\mathbb{D}}_{\Grave{\Grave{t}}}^{N}$ = $\Check{\mathbb{D}}^{N}$; b) $\Check{\mathbb{D}}_{\grave{t}}^{N}$ $\cap$ $\Check{\mathbb{D}}_{\Grave{\Grave{t}}}^{N}$ = $\phi$; c) $n^{**}$, $n^{***}$ $\leq$ $n$; and d) $n^{**}$ = $n \times z$, and $n^{***}$ = $n \times (1 - z)$, where $z$ belongs to $\mathbb{Z}$ or the value of $z$ can lie between 0 and 1 for $\textit{$CM^{N}$}$. 
 	\begin{figure*}[!htbp]
	\begin{center}
	\includegraphics[width=0.85\textwidth]{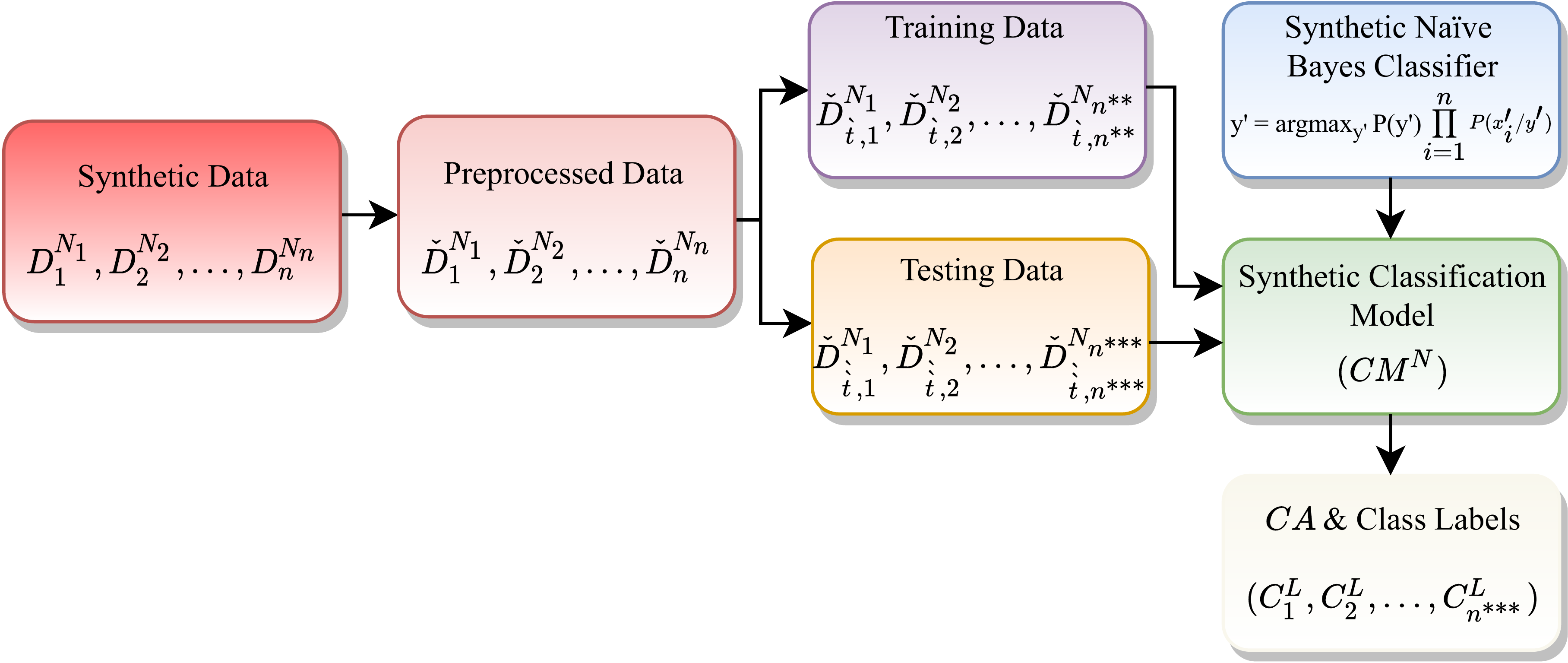}
		\caption{Classification flow for shared data and model}
		\label{fig:FC45}
	\end{center}
\end{figure*}
The utilization of training data $\Check{D}_{\grave{t},1}^{N_{1}}$, $\Check{D}_{\grave{t},2}^{N_{2}}$, $\dots$, $\Check{D}_{\grave{t},n^{**}}^{N_{n^{**}}}$ is to train $\textit{$CM^{N}$}$, whereas the accuracy of $\textit{$CM^{N}$}$ is assessed by the testing data $\Check{D}_{\Grave{\Grave{t}},1}^{N_{1}}$, $\Check{D}_{\Grave{\Grave{t}},2}^{N_{2}}$, $\dots$, $\Check{D}_{\Grave{\Grave{t}},n^{***}}^{N_{n^{***}}}$. Data objects $\Check{D}_{\Grave{\Grave{t}},1}^{N_{1}}$, $\Check{D}_{\Grave{\Grave{t}},2}^{N_{2}}$, $\dots$, $\Check{D}_{\Grave{\Grave{t}},n^{***}}^{N_{n^{***}}}$ are provided to $\textit{$CM^{N}$}$ in the testing process to determine the classes. $\textit{$CM^{N}$}$ examines $\Check{D}_{\Grave{\Grave{t}},1}^{N_{1}}$, $\Check{D}_{\Grave{\Grave{t}},2}^{N_{2}}$, $\dots$, $\Check{D}_{\Grave{\Grave{t}},n^{***}}^{N_{n^{***}}}$, and generates the output as the class label vector $\mathbb{C}^{L}$ = \{$C_{1}^{L}$, $C_{2}^{L}$, $\dots$, $C_{n^{***}}^{L}$\}. Using $C_{1}^{L}$, $C_{2}^{L}$, $\dots$, $C_{n^{***}}^{L}$, the Classification Accuracy (\textit{CA}) of $\textit{$CM^{N}$}$ is measured by applying the Eq. (9). In this Eq., $C_{1}^{L}$, $C_{2}^{L}$, $\dots$, $C_{n^{'''}}^{L}$ indicates the number of correctly identified samples and $C_{1}^{L}$, $C_{2}^{L}$, $\dots$, $C_{n^{***}}^{L}$ is the total number of sample in test data, where $n^{'''}$ $\in$ [1, $n^{***}$]. 
\begin{equation}
    CA = \frac{C_{1}^{L}, C_{2}^{L}, \dots, C_{n^{'''}}^{L}}{C_{1}^{L}, C_{2}^{L}, \dots, C_{n^{***}}^{L}}
\end{equation}
The precision (\textit{P}), and recall (\textit{R}) are computed using Eq. (10) and (11), respectively.
\begin{equation}
    P = \frac{(C_{1}^{L}, C_{2}^{L}, \dots, C_{n^{'''}}^{L}) \cap (C_{1}^{L}, C_{2}^{L}, \dots, C_{n^{***}}^{L})}{C_{1}^{L}, C_{2}^{L}, \dots, C_{n^{'''}}^{L}}
\end{equation}
\begin{equation}
    R = \frac{(C_{1}^{L}, C_{2}^{L}, \dots, C_{n^{'''}}^{L}) \cap (C_{1}^{L}, C_{2}^{L}, \dots, C_{n^{***}}^{L})}{C_{1}^{L}, C_{2}^{L}, \dots, C_{n^{***}}^{L}}
\end{equation}
The F1-Score (\textit{FS}) is measured using Eq. (12).
\begin{equation}
    FS = 2 \times \frac{ (P \times R)}{(P+R)}
\end{equation}
\section{Illustration}
Let's assume that the proposed model consists of twenty owners: \textit{$DO_{id}$} = \{$DO_{1}$, $DO_{2}$, $\dots$, $DO_{20}$\} possess data \textit{$D$} = \{$D_{1}$, $D_{2}$, $\dots$, $D_{20}$\}. Data $D_{1}$ contains the information (4.2575, 7.1754, -1.7685, -2.6898, 0), $D_{2}$ contains the information (0.7536, 0.9137, 6.1110, -4.0517, 0), $\dots$, $D_{20}$ contains the information (6.2122, 7.6552, -2.4515, -1.4451, 0). Before sending it to \textit{$CSP$}, $DO_{1}$, $DO_{2}$, $\dots$, $DO_{20}$ generate the noise $N_{1}$, $N_{2}$, $\dots$, $N_{20}$ by applying the Eqs. (1) to (4). The obtained noise $N_{1}$ incorporates 0.6569, $N_{2}$ incorporates 0.8015, $\dots$, $N_{20}$ incorporates 0.6959. Table 3 shows the synthetic data (\textit{SD}) that are produced after injecting $N_{1}$, $N_{2}$, $\dots$, $N_{20}$ into actual data (\textit{AD}).
\begin{table}[!htbp]
\setcounter{table}{2}
\centering
	\caption{Synthetic data v/s Actual data}
		\label{TableExp64}
\resizebox{\textwidth}{!}{\begin{tabular}{|c|c|c|c|c|c|c|c|c|c|}\hline 
\multirow{2}{*}{Data Owners} & \multicolumn{2}{ c| }{$\Lambda_{1}$} & \multicolumn{2}{ c| }{$\Lambda_{2}$} & \multicolumn{2}{ c|}{$\Lambda_{3}$} & \multicolumn{2}{ c|}{$\Lambda_{4}$} & \multirow{2}{*}{$\Lambda_{5}$} \\ \cline{2-9}
& $SD_{1}$ & $AD_{1}$ & $SD_{2}$ & $AD_{2}$ & $SD_{3}$ & $AD_{3}$ & $SD_{4}$ & $AD_{4}$ & \\  \hline
$DO_{1}$ &  4.9144 & 4.2575 & 7.8223 & 7.1754 & -1.1116 & -1.7685 & -2.0329 & -2.6898   & 0\\ \hline
$DO_{2}$ & 1.5551 & 0.7536 & 1.7152 & 0.9137 & 6.9125 & 6.1110 & -3.2502 &  -4.0517 & 0\\ \hline
$DO_{3}$ & 1.3112 & 1.0978 & -2.6567 & -2.8701 & 2.8156 & 2.6022 & -1.8123 & -2.0257 & 0 \\ \hline
$DO_{4}$ & 2.9124 & 2.1312 & -2.3231 & -3.1043 & 3.4122 & 2.631 & -2.9158 & -3.6970 &  1 \\ \hline
$DO_{5}$ & 5.1451 & 4.9639 & 3.4256 & 3.2444 & -1.6823 & -1.8635 & 1.2151 & 1.0339 & 1 \\ \hline
$DO_{6}$ & 1.0132 & 0.7983 & 5.7588 & 5.5439 & 0.4724 & 0.2575 & -0.6125 & -0.8274 & 0 \\ \hline
$DO_{7}$ & 2.5156 & 1.9033 & 9.1772 & 8.5649 & -0.7336 & -1.3459 & -0.7353 & -1.3476 & 0 \\ \hline
$DO_{8}$ & 4.4157 & 4.2046 & 8.7779 & 8.5668 & -4.4035 & -4.6146 & -0.8064 & -1.0175 & 0 \\ \hline
$DO_{9}$ & 1.9877 & 1.6659 & 1.1038 & 0.7820 & 2.3946 & 2.0728 & 0.8629 & 0.5411 & 1 \\ \hline
$DO_{10}$ & 2.9935 & 2.5039 & 7.6625 & 7.1729 & 0.1539 & -0.3357 & -1.0175 & -1.5071 & 1 \\ \hline
$DO_{11}$ & -2.6128  & -2.9317 & 10.8430 & 10.5241 & 2.5462 & 2.2273 & -2.7926 & -3.1115 & 0 \\ \hline
$DO_{12}$ & 2.4121 & 2.012 & 8.7261 & 8.3260 & -3.0030 & -3.4031 & -0.5724 & -0.9725 & 0 \\ \hline
$DO_{13}$ & 3.6651 & 2.7939 & -3.3924 & -4.2636 & 3.4896 & 2.6184 & 1.4771 & 0.6059 & 0 \\ \hline
$DO_{14}$ & 1.7219 & 1.5800& 3.0646 & 2.927 & 0.8747 & 0.7328 & 0.5861 & 0.4442 & 1 \\ \hline
$DO_{15}$ & 1.4512 & 1.2401 & 2.8473 & 2.6362 & 4.3439 & 4.1328 & 1.6524 &  1.4413 & 1 \\ \hline
$DO_{16}$ & 2.7159 & 2.5497 & -4.0257 & -4.1919 & 8.3428 & 8.1766 & -2.1086 & -2.2748 & 0 \\ \hline
$DO_{17}$ & 3.2153 & 2.7866 & 11.0272 & 10.5985 & -2.3564 & -2.7851 & -2.1113 & -2.5400 & 0 \\ \hline
$DO_{18}$ & 1.7121 & 1.5903 & 7.8902 & 7.7684 & -1.3663 & -1.4881 & -1.5650 & -1.6868 & 0 \\ \hline
$DO_{19}$ & 1.2192 & 0.9812 & -1.1117 & -1.3496 & 2.7118 & 2.4739 & -1.4445 & -1.6824 & 1 \\ \hline
$DO_{20}$ & 6.9081 & 6.2122 & 8.3511 & 7.6552 & -1.7556 & -2.4515 & -0.7492 & -1.4451 & 1 \\ \hline
\end{tabular}}
\end{table}
To preserve the privacy of \textit{CM}, \textit{$CO$} also generates the noise $\hat{N}_{1}$, $\hat{N}_{2}$ using Eqs. (1) to (4). It is supposed that the value of $\hat{N}_{1}$ is 0.5654, and $\hat{N}_{2}$ is 0.3651. Initially, \textit{$CO$} measures the sensitivity of mean ($\mu$), and standard deviation ($\sigma$) by adopting Eq. (5). It is observed that the value of $\mu$ is 16.3258, and the value of $\sigma$ is 9.6894. Afterward, $\hat{N}_{1}$, and $\hat{N}_{2}$ are injected into $\mu$ and $\sigma$, respectively, by applying Eqs. (6) and (7). The synthetic model \textit{$CM^{N}$}, and \textit{SD} are uploaded to cloud platform. \textit{$CSP$} preprocesses \textit{SD} using Eq. (8) and acquires the preprocessed data as $\Check{D}^{N}$ as (1.2456, 2.4523, -0.4452, 1.3425), (-2.5434, -0.8145, 1.9315, 1.3155), $\dots$, (-1.3141, 2.5145, -1.3145, -1.0485). The \textit{$CM^{N}$} accepts 30 input data, and provides 30 output labels. The \textit{CA} of \textit{$CM^{N}$} is calculated by using Eq. (9). If \textit{$CM^{N}$} identifies 18 correct labels then \textit{$CM^{N}$}'s \textit{CA} is 60\% ((18/30) * 100), where 30 is the total number of test items. The \textit{P} of $\textit{$CM^{N}$}$ is measured by applying Eq. (10). If $C_{1}^{L}$, $C_{2}^{L}$, $\dots$, $C_{n^{'''}}^{L}$ is 10, $C_{1}^{L}$, $C_{2}^{L}$, $\dots$, $C_{n^{***}}^{L}$ is 30, and common labels among them is 6 then \textit{$CM^{N}$}'s \textit{P} is 60\% ((6/10) * 100). The \textit{R} of \textit{$CM^{N}$} is computed by adopting Eq. (11). The \textit{$CM^{N}$}'s \textit{R} is 20\% ((6/30) * 100). The \textit{FS} of \textit{$CM^{N}$} is estimated by using Eq. (12). The \textit{$CM^{N}$}'s \textit{FS} is 30\% (2 * (\textit{P} * \textit{R}) / (\textit{P} + \textit{R})), where \textit{P} is 60\%, and \textit{R} is 20\%.    
\section{Operational Design and Computational Complexity}
The operational summary of the proposed model is described by Algorithm 1, which protects the data as well as the model and performs the classification task. Initially, the list of data ($D$), noise ($N$), and preprocessed data $\Check{D}^{N}$ are initialized. To preserve the privacy of the owner's data $D$, $N$ is generated. The obtained $N$ is used to construct the synthetic data after injecting it into $D$. To protect the classification model, the sensitivity of $\mu$ is calculated. Based on this sensitivity, the scale factor is measured for $\mu$. Using the scale factor, the noise is generated and adds it into $\mu$. Similarly, the sensitivity of $\sigma$ is also computed, and the scale factor is determined using this sensitivity. 
\begin{algorithm}
\KwIn{Actual data $D$, Noise vector $N$, Privacy parameter $\epsilon$}
\KwOut{unknown class label $C_{i}^{L}$, \textit{CA}, \textit{P}, \textit{R}, and \textit{FS}}
Initialize data $D$ := \{$D_{1}$, $D_{2}$, $\dots$, $D_{n}$\}, $N$ := \{$N_{1}$, $N_{2}$, $\dots$, $N_{n}$\}, $\Check{D}_{1}^{N}$ = \{$(x_{1}, y_{1})$, $(x_{2}, y_{2})$, $\dots$, $(x_{\ddot{i}}, y_{\ddot{j}})$\}  \\
\textbf{for} $i$ = $1$, $2$, $\dots$, $n$ \textbf{do}\\
\hspace{.3cm} $N_{i}$ = Lap(0,1) \\
\hspace{.3cm}$D_{i}^{N}$ = $D_{i}$ + $N_{i}$ \\
\hspace{.3cm}\textbf{for} each attribute $X_{j}$ \textbf{do}\\
\hspace{.6cm} Compute the sensitivity $S^{\mu}$ with mean $\mu$ with bound [$g_{i}$, $h_{i}$] \\
\hspace{.6cm} Scale factor for $\mu$ ($Sf^{\mu}$) $\leftarrow$ $S^{\mu}$ / $\epsilon$ \\
\hspace{.6cm} $\mu_{i}^{'}$ = $\mu_{i}$ + Lap(0, $Sf^{\mu}$) \\ 
\hspace{.6cm} Compute the sensitivity $S^{\sigma}$ with standard deviation $\sigma$ with bound [$g_{i}$, $h_{i}$] \\
\hspace{.6cm} Scale factor for $\sigma$ ($Sf^{\sigma}$) $\leftarrow$ $S^{\sigma}$ / $\epsilon$ \\
\hspace{.6cm} $\sigma_{i}^{'}$ = $\sigma_{i}$ + Lap(0, $Sf^{\sigma}$) \\ 
\hspace{.6cm} Use $\mu_{i}^{'}$, and $\sigma_{i}^{'}$ to compute $P(\mathbb{D}^{N}=x_{j}|\mathbb{C}=C_{i})$ \\
\hspace{.3cm}\textbf{end for} \\ 
\hspace{.3cm}\textbf{for} each class $c_{i}$ \textbf{do}\\
\hspace{.6cm} Count $nc_{i}^{'}$ $\leftarrow$ $nc_{i}$ + Lap (0,1) \\
\hspace{.6cm} Use $nc_{i}^{'}$ to compute the prior $P(c_{i})$ \\
\hspace{.3cm}\textbf{end for} \\ 
\hspace{.3cm} $C_{i}^{L}$ = $argmax_{c}$ $P(C)$ $\prod$ $P(X_{j}\mid c)$\\
\textbf{end for} \\
\textit{CA} = (\#$Correctly$ $classified$ $sample$ /  \#$test$ $sample$) * 100 \\
\textit{P} = \{relevant data labels $\cap$ retrieved data labels\} / \{retrieved data labels\}\\
\textit{R} = \{relevant data labels $\cap$ retrieved data labels\} / \{relevant data labels\}\\
\textit{FS} = $2 * (\textit{P} * \textit{R}) / (\textit{P} + \textit{R})$
\caption{DA-PMLM operational summary}
\label{algo:bb}
\end{algorithm}
Based on this scale factor, noise is produced and added into $\sigma$. The synthetic $\mu$, and $\sigma$ are used to compute $P(\mathbb{D}^{N}=x_{j}|\mathbb{C}=C_{i})$. The test data is given and received class labels from the classification model. Using these class labels, the classification parameter values are measured in a privacy-preserving manner.
\par
In algorithm 1, steps 2 to 19 perform the classification over synthetic data using a synthetic model, whose time complexity depends on the noise generation, noise addition, and classifier use. To preserve the privacy of the data, noise is generated in step 3 and added using the Laplace mechanism in step 4, which takes $\mathcal{O}(n^{2})$ time, where $n$ is the total rows in the dataset. The classification model is protected using steps 5 to 18. The \textit{NB} classifier takes $\mathcal{O}(ni)$ time. The computation time for the sensitivity of $\mu$ and $\sigma$ is $\mathcal{O}(n^{2})$ time. Steps 20 to 23 require $\mathcal{O}(1)$ time. Therefore, the total time complexity of DA-PMLM is $\mathcal{O}(n^{2})$.
\section{Performance Evaluation}
\subsection{Experimental Setup}
The experiments are conducted on a system equipped with Intel (R) Core (TM) i5-4210U CPU and 3.60 GHz clock speed. The computation machine is deployed with 64-bit Ubuntu along with 8 GB of main memory. The classification tasks are accomplished using Python 2.7.15 programming language. The \textit{NB} classifier has been used over testing data. 
\subsection{Datasets and Classification Parameters}
 Heart Disease, Iris, Balance Scale, and Nursery datasets, all of which were obtained from the UCI Machine Learning Repository \cite{song2012automatic} to train \textit{$CM$}. These datasets have 75, 5, 5, 9 attributes and 303, 150, 625, 12960 instances. The description of the datasets is shown in Table 4. 
 \par
 The 9/10 of the data from the complete dataset is utilized as training data, while the remaining is used as test data for the purpose of training the \textit{CM}. The machine learning tasks are accomplished over the clean data, DA-PMLM, PMLM \cite{li2018privacy}, MLPAM \cite{gupta2020mlpam}, DPEL \cite{LI202134}, and PDLM \cite{ma2018pdlm}. We have used the Laplace mechanism to generate the noise. However, DA-PMLM, PMLM \cite{li2018privacy}, MLPAM \cite{gupta2020mlpam}, as well as DPEL \cite{LI202134} schemes contain noise, and PDLM \cite{ma2018pdlm} scheme is based on a homomorphic encryption technique. The results of \textit{CM} are evaluated using test data, and \textit{CA}, \textit{P}, \textit{R}, and \textit{FS} are computed from these results. 
\begin{table}[!htbp]
\centering
		\caption{Basic information of four datasets}
		\label{TableExp17}
			\begin{tabular}{|c|c|c|c|c|c|}\hline 
Dataset  & \#Instances & \#Features & \#Classes & Samples in  & Samples in  \\ 
  &  &  &  &  training set &  test set\\ \hline
Heart Disease & 303 & 75 & 2 & 272 & 31\\ \hline
Iris & 150 & 5 & 3 & 135 & 15\\ \hline
Balance Scale & 625 & 5 & 3 & 562 & 63\\ \hline
Nursery  & 12960 & 9 & 5 & 11664 & 1296 \\ \hline
\end{tabular}
\end{table}
\subsection{Results}
The \textit{$CM$} acquires the classification results including \textit{CA}, \textit{P}, \textit{R}, and \textit{FS} over the clean data, DA-PMLM, PMLM \cite{li2018privacy}, MLPAM \cite{gupta2020mlpam}, DPEL \cite{LI202134}, and PDLM \cite{ma2018pdlm} as shown in Figs. 4(a)-(d) to 7(a)-(d). In DA-PMLM, the maximal value of \textit{CA} is 93.33\% on the Iris dataset with 2.0 privacy budget. The minimal value of \textit{CA} is 33.87\% on the Nursery dataset with 0.01 privacy budget. The average value of \textit{CA} is 70.04\%, 74.97\%, 61.44\%, and 46.28\% over Heart Disease, Iris, Balance Scale, and Nursery datasets, respectively. The greatest value of \textit{P} is 94.44\% on the Iris dataset with 2.0 privacy budget. The lowest value of \textit{P} is 31.57\% on the Nursery dataset with 0.01 privacy budget. The average value of \textit{P} is 70.53\%, 72.67\%, 59.66\%, and 44.89\% over Heart Disease, Iris, Balance Scale, and Nursery datasets, respectively. The maximum value of \textit{R} is 93.33\% on the Iris dataset with 2.0 privacy budget.
\begin{figure}[!htbp]
\centering
\begin{subfigure}[t]{0.49\textwidth}
\begin{tikzpicture}[node distance = 1cm,auto,scale=.70, transform shape]
\pgfplotsset{every axis y label/.append style={rotate=180,yshift=10.5cm}}
\begin{axis}[
      axis on top=false,
      xmin=15, xmax=145,
      ymin=0, ymax=0.9,
      xtick={22,42,62,80,100,120,138},
      xticklabels={0.01,0.05, 0.1,0.5,1.0,1.5,2.0},
        ycomb,
        ylabel near ticks, yticklabel pos=left,
      ylabel={Accuracy},
      xlabel={Privacy Budget},
      legend style={at={(0.5,-0.18)},
      anchor=north,legend columns=2},
      ymajorgrids=true,
      grid style=dashed,
          ]
\addplot+[mark options={fill=blue},fill=blue!40!,draw=blue,  thick]
coordinates
{(17,.8387) (36,.8387) (55,.8387) (74,.8387) (93,.8387) (112,.8387) (131,.8387)}
\closedcycle;%
\addlegendentry{Clean Data}

\addplot+[mark options={fill=green},fill=pink,draw=green,  thick] 
coordinates
 {(20,.6129) (39,.6229) (58,.6885) (77,.7377) (96,.7540) (115,.7868) (134,.8196)}
\closedcycle;%
\addlegendentry{PMLM \cite{li2018privacy}, MLPAM \cite{gupta2020mlpam}}

\addplot+[mark options={fill=white},fill=red!60!,draw=red!70!,  thick] 
coordinates
 {(23,.6129) (42,.6451) (61,.6774) (80,.7096) (99,.7419) (118,.7741) (137,.8064)}
\closedcycle;%
\addlegendentry{DPEL \cite{LI202134}}
\addplot+[mark options={fill=teal},fill=teal,draw=teal,  thick] 
coordinates
 {(26,.6777) (45,.6777) (64,.6777) (83,.6777) (102,.6777) (121,.6777) (140,.6777)}
\closedcycle;%
\addlegendentry{\hspace{-1.75cm} PDLM \cite{ma2018pdlm}}
\addplot+[mark options={fill=cyan},fill=cyan,draw=cyan,  thick] 
coordinates
 {(29,.5806) (48,.6129) (67,.6774) (86,.7096) (105,.7419) (124,.7419) (143,.7419)}
\closedcycle;%
\addlegendentry{DA-PMLM}
\end{axis}
\end{tikzpicture}
            \caption{Heart Disease}
        \end{subfigure}
                  \hfill
\begin{subfigure}[t]{0.49\textwidth}
\begin{tikzpicture}[node distance = 1cm,auto,scale=.70, transform shape]
\pgfplotsset{every axis y label/.append style={rotate=180,yshift=10.5cm}}
\begin{axis}[
      axis on top=false,
      xmin=15, xmax=145,
      ymin=0, ymax=1.1,
      xtick={22,42,62,80,100,120,138},
      xticklabels={0.01,0.05,0.1,0.5,1.0,1.5,2.0},
        ycomb,
        ylabel near ticks, yticklabel pos=left,
      ylabel={Accuracy},
      xlabel={Privacy Budget},
      legend style={at={(0.5,-0.18)},
      anchor=north,legend columns=2},
      ymajorgrids=true,
      grid style=dashed,
          ]
\addplot+[mark options={fill=blue},fill=blue!40!,draw=blue,  thick]
coordinates
{(17,.9333) (36,.9333) (55,.9333) (74,.9333) (93,.9333) (112,.9333) (131,.9333)}
\closedcycle;%
\addlegendentry{Clean Data}

\addplot+[mark options={fill=green},fill=pink,draw=green,  thick] 
coordinates
 {(20,.6000) (39,.6666) (58,.8000) (77,.8247) (96,.8666) (115,.9333) (134,1.0)}
\closedcycle;%
\addlegendentry{PMLM \cite{li2018privacy}, MLPAM \cite{gupta2020mlpam}}

\addplot+[mark options={fill=white},fill=red!60!,draw=red!70!,  thick] 
coordinates
 {(23,.6000) (42,.6666) (61,.7333) (80,.8000) (99,.8666) (118,.9333) (137,1.0)}
\closedcycle;%
\addlegendentry{DPEL \cite{LI202134}}

\addplot+[mark options={fill=teal},fill=teal,draw=teal,  thick] 
coordinates
 {(26,.7638) (45,.7638) (64,.7638) (83,.7638) (102,.7638) (121,.7638) (140,.7638)}
\closedcycle;%
\addlegendentry{\hspace{-1.75cm} PDLM \cite{ma2018pdlm}}

\addplot+[mark options={fill=cyan},fill=cyan,draw=cyan,  thick] 
coordinates
 {(29,.5333) (48,.6000) (67,.7333) (86,.7815) (105,.8000) (124,.8666) (143,.9333)}
\closedcycle;%
\addlegendentry{DA-PMLM }
\end{axis}
\end{tikzpicture}
            \caption{Iris}
        \end{subfigure}
                  \hfill
\begin{subfigure}[t]{0.49\textwidth}
\begin{tikzpicture}[node distance = 1cm,auto,scale=.70, transform shape]
\pgfplotsset{every axis y label/.append style={rotate=180,yshift=10.5cm}}
\begin{axis}[
      axis on top=false,
      xmin=15, xmax=145,
      ymin=0, ymax=0.9,
      xtick={22,42,62,80,100,120,138},
      xticklabels={0.01,0.05,0.1,0.5,1.0,1.5,2.0},
        ycomb,
        ylabel near ticks, yticklabel pos=left,
      ylabel={Accuracy},
      xlabel={Privacy Budget},
      legend style={at={(0.5,-0.18)},
      anchor=north,legend columns=2},
      ymajorgrids=true,
      grid style=dashed,
          ]
\addplot+[mark options={fill=blue},fill=blue!40!,draw=blue,  thick]
coordinates
{(17,.8571) (36,.8571) (55,.8571) (74,.8571) (93,.8571) (112,.8571) (131,.8571)}
\closedcycle;%
\addlegendentry{Clean Data}

\addplot+[mark options={fill=green},fill=pink,draw=green,  thick] 
coordinates
 {(20,.5396) (39,.5714) (58,.6013) (77,.6507) (96,.6984) (115,.7311) (134,.7936)}
\closedcycle;%
\addlegendentry{PMLM \cite{li2018privacy}, MLPAM \cite{gupta2020mlpam}}

\addplot+[mark options={fill=white},fill=red!60!,draw=red!70!,  thick] 
coordinates
 {(23,.5555) (42,.5873) (61,.6349) (80,.6666) (99,.6984) (118,.7301) (137,.8095)}
\closedcycle;%
\addlegendentry{DPEL \cite{LI202134}}

\addplot+[mark options={fill=teal},fill=teal,draw=teal,  thick] 
coordinates
 {(26,.7293) (45,.7293) (64,.7293) (83,.7293) (102,.7293) (121,.7293) (140,.7293)}
\closedcycle;%
\addlegendentry{\hspace{-1.75cm} PDLM \cite{ma2018pdlm}}

\addplot+[mark options={fill=cyan},fill=cyan,draw=cyan,  thick] 
coordinates
 {(29,.4920) (48,.5238) (67,.5714) (86,.6349) (105,.6666) (124,.6825) (143,.7301)}
\closedcycle;%
\addlegendentry{DA-PMLM}
\end{axis}
\end{tikzpicture}
            \caption{Balance}
        \end{subfigure}
                  \hfill
\begin{subfigure}[t]{0.49\textwidth}
\begin{tikzpicture}[node distance = 1cm,auto,scale=.70, transform shape]
\pgfplotsset{every axis y label/.append style={rotate=180,yshift=10.5cm}}
\begin{axis}[
      axis on top=false,
      xmin=15, xmax=145,
      ymin=0, ymax=0.8,
      xtick={22,42,62,80,100,120,138},
      xticklabels={0.01,0.05,0.1,0.5,1.0,1.5,2.0},
        ycomb,
        ylabel near ticks, yticklabel pos=left,
      ylabel={Accuracy},
      xlabel={Privacy Budget},
      legend style={at={(0.5,-0.18)},
      anchor=north,legend columns=2},
      ymajorgrids=true,
      grid style=dashed,
          ]
\addplot+[mark options={fill=blue},fill=blue!40!,draw=blue,  thick]
coordinates
{(17,.7615) (36,.7615) (55,.7615) (74,.7615) (93,.7615) (112,.7615) (131,.7615)}
\closedcycle;%
\addlegendentry{Clean Data}

\addplot+[mark options={fill=green},fill=pink,draw=green,  thick] 
coordinates
 {(20,.3495) (39,.3719) (58,.3888) (77,.4436) (96,.5979) (115,.6875) (134,.7438)}
\closedcycle;%
\addlegendentry{PMLM \cite{li2018privacy}, MLPAM \cite{gupta2020mlpam}}

\addplot+[mark options={fill=white},fill=red!60!,draw=red!70!,  thick] 
coordinates
 {(23,.3518) (42,.3827) (61,.4012) (80,.5169) (99,.6250) (118,.6628) (137,.7229)}
\closedcycle;%
\addlegendentry{DPEL \cite{LI202134}}

\addplot+[mark options={fill=teal},fill=teal,draw=teal,  thick] 
coordinates
 {(26,.6703) (45,.6703) (64,.6703) (83,.6703) (102,.6703) (121,.6703) (140,.6703)}
\closedcycle;%
\addlegendentry{\hspace{-1.75cm} PDLM \cite{ma2018pdlm}}

\addplot+[mark options={fill=cyan},fill=cyan,draw=cyan,  thick] 
coordinates
 {(29,.3387) (48,.3418) (67,.3536) (86,.3834) (105,.5054) (124,.6057) (143,.7121)}
\closedcycle;%
\addlegendentry{DA-PMLM }
\end{axis}
\end{tikzpicture}
            \caption{Nursery}
        \end{subfigure}
\caption{Accuracy of \textit{CM} in DA-PMLM}
\end{figure}
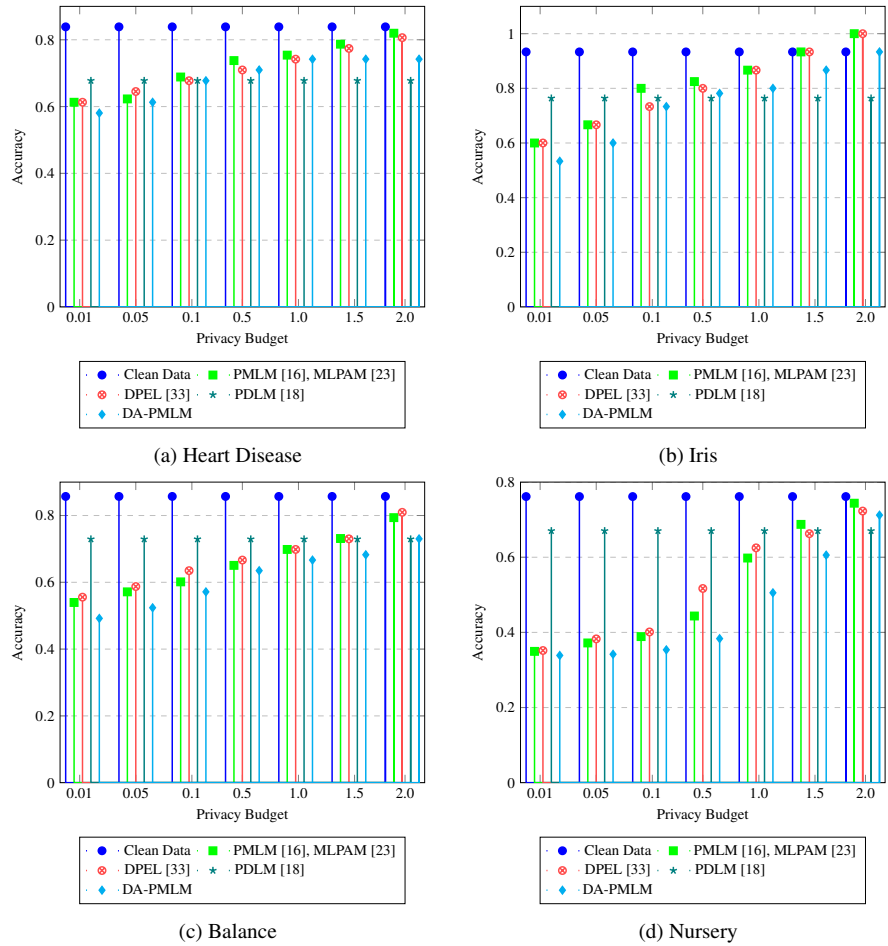 
 The minimum value of \textit{R} is 33.87\% on the Nursery dataset with 0.01 privacy budget. The average value of \textit{R} is 70.04\%, 74.97\%, 61.44\%, and 46.28\% over Heart Disease, Iris, Balance Scale, and Nursery datasets, respectively. The highest value of \textit{FS} is 93.33\% on the Iris dataset with 2.0 privacy budget. The lowest value of \textit{FS} is 28.66\% on the Nursery dataset with 0.01 privacy budget. The average value of \textit{FS} is 69.30\%, 73.13\%, 58.43\%, and 43.64\%  over Heart Disease, Iris, Balance Scale, and Nursery datasets, respectively. The performance of the datasets decreases in order: Iris $>$ Heart Disease $>$ Balance Scale $>$ Nursery datasets. 
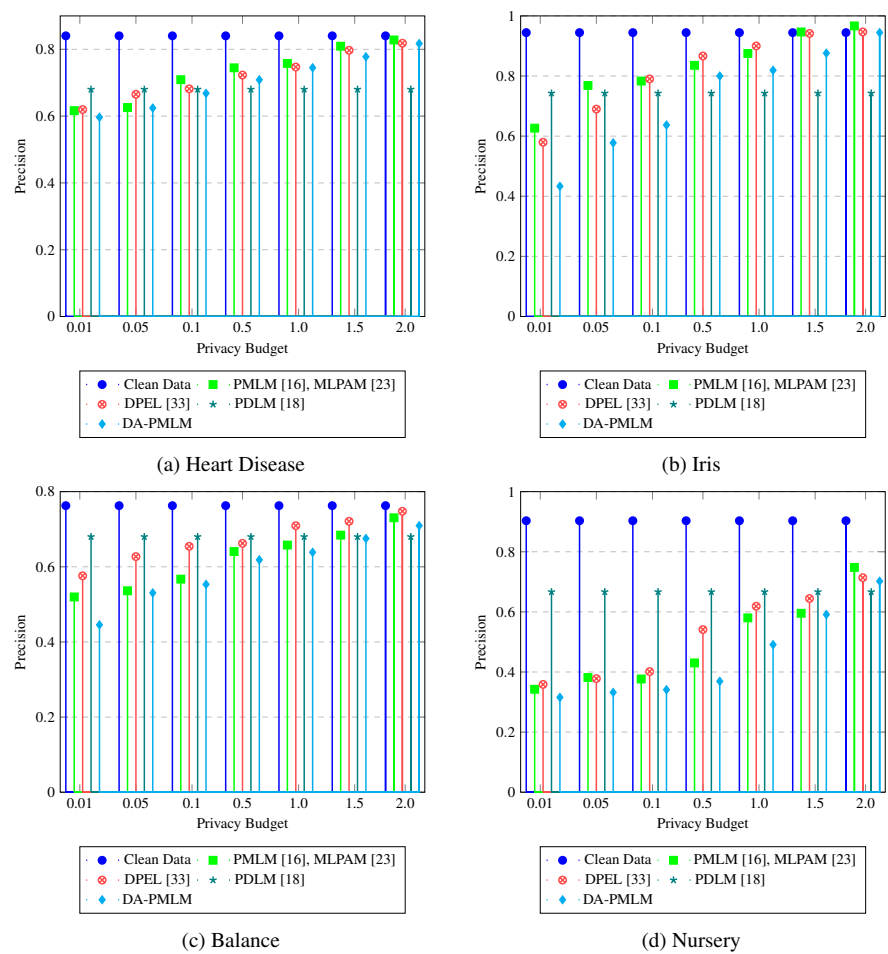
\begin{figure}[!htbp]
\centering
\begin{subfigure}[t]{0.49\textwidth}
\begin{tikzpicture}[node distance = 1cm,auto,scale=.70, transform shape]
\pgfplotsset{every axis y label/.append style={rotate=180,yshift=10.5cm}}
\begin{axis}[
      axis on top=false,
      xmin=15, xmax=145,
      ymin=0, ymax=0.9,
      xtick={22,42,62,80,100,120,138},
      xticklabels={0.01,0.05,0.1,0.5,1.0,1.5,2.0},
        ycomb,
        ylabel near ticks, yticklabel pos=left,
      ylabel={Precision},
      xlabel={Privacy Budget},
      legend style={at={(0.5,-0.18)},
      anchor=north,legend columns=2},
      ymajorgrids=true,
      grid style=dashed,
          ]
\addplot+[mark options={fill=blue},fill=blue!40!,draw=blue,  thick]
coordinates
{(17,.8404) (36,.8404) (55,.8404) (74,.8404) (93,.8404) (112,.8404) (131,.8404)}
\closedcycle;%
\addlegendentry{Clean Data}

\addplot+[mark options={fill=green},fill=pink,draw=green,  thick] 
coordinates
 {(20,.6164) (39,.6258) (58,.7091) (77,.7446) (96,.7578) (115,.8093) (134,.8279)}
\closedcycle;%
\addlegendentry{PMLM \cite{li2018privacy}, MLPAM \cite{gupta2020mlpam}}

\addplot+[mark options={fill=white},fill=red!60!,draw=red!70!,  thick] 
coordinates
 {(23,.6196) (42,.6653) (61,.6817) (80,.7230) (99,.7470) (118,.7971) (137,.8183)}
\closedcycle;%
\addlegendentry{DPEL \cite{LI202134}}

\addplot+[mark options={fill=teal},fill=teal,draw=teal,  thick] 
coordinates
 {(26,.6800) (45,.6800) (64,.6800) (83,.6800) (102,.6800) (121,.6800) (140,.6800)}
\closedcycle;%
\addlegendentry{\hspace{-1.75cm} PDLM \cite{ma2018pdlm}}

\addplot+[mark options={fill=cyan},fill=cyan,draw=cyan,  thick] 
coordinates
 {(29,.5967) (48,.6242) (67,.6683) (86,.7087) (105,.7449) (124,.7777) (143,.8172)}
\closedcycle;%
\addlegendentry{ DA-PMLM }
\end{axis}
\end{tikzpicture}
            \caption{Heart Disease}
        \end{subfigure}
                  \hfill
\begin{subfigure}[t]{0.49\textwidth}
\begin{tikzpicture}[node distance = 1cm,auto,scale=.70, transform shape]
\pgfplotsset{every axis y label/.append style={rotate=180,yshift=10.5cm}}
\begin{axis}[
      axis on top=false,
      xmin=15, xmax=145,
      ymin=0, ymax=1,
      xtick={22,42,62,80,100,120,138},
      xticklabels={0.01,0.05,0.1,0.5,1.0,1.5,2.0},
        ycomb,
        ylabel near ticks, yticklabel pos=left,
      ylabel={Precision},
      xlabel={Privacy Budget},
      legend style={at={(0.5,-0.18)},
      anchor=north,legend columns=2},
      ymajorgrids=true,
      grid style=dashed,
          ]
\addplot+[mark options={fill=blue},fill=blue!40!,draw=blue,  thick]
coordinates
{(17,.9444) (36,.9444) (55,.9444) (74,.9444) (93,.9444) (112,.9444) (131,.9444)}
\closedcycle;%
\addlegendentry{Clean Data}

\addplot+[mark options={fill=green},fill=pink,draw=green,  thick] 
coordinates
 {(20,.6266) (39,.7685) (58,.7833) (77,.8355) (96,.8750) (115,.9466) (134,.9666)}
\closedcycle;%
\addlegendentry{PMLM \cite{li2018privacy}, MLPAM \cite{gupta2020mlpam}}

\addplot+[mark options={fill=white},fill=red!60!,draw=red!70!,  thick] 
coordinates
 {(23,.5797) (42,.6900) (61,.7904) (80,.8666) (99,.9000) (118,.9416) (137,.9466)}
\closedcycle;%
\addlegendentry{DPEL \cite{LI202134}}

\addplot+[mark options={fill=teal},fill=teal,draw=teal,  thick] 
coordinates
 {(26,.7428) (45,.7428) (64,.7428) (83,.7428) (102,.7428) (121,.7428) (140,.7428)}
\closedcycle;%
\addlegendentry{\hspace{-1.75cm} PDLM \cite{ma2018pdlm}}

\addplot+[mark options={fill=cyan},fill=cyan,draw=cyan,  thick] 
coordinates
 {(29,.4333) (48,.5777) (67,.6370) (86,.7999) (105,.8191) (124,.8761) (143,.9444)}
\closedcycle;%
\addlegendentry{DA-PMLM }
\end{axis}
\end{tikzpicture}
            \caption{Iris}
        \end{subfigure}
                  \hfill
\begin{subfigure}[t]{0.49\textwidth}
\begin{tikzpicture}[node distance = 1cm,auto,scale=.70, transform shape]
\pgfplotsset{every axis y label/.append style={rotate=180,yshift=10.5cm}}
\begin{axis}[
      axis on top=false,
      xmin=15, xmax=145,
      ymin=0, ymax=0.8,
      xtick={22,42,62,80,100,120,138},
      xticklabels={ 0.01,0.05,0.1,0.5,1.0,1.5,2.0},
        ycomb,
        ylabel near ticks, yticklabel pos=left,
      ylabel={Precision},
      xlabel={Privacy Budget},
      legend style={at={(0.5,-0.18)},
      anchor=north,legend columns=2},
      ymajorgrids=true,
      grid style=dashed,
          ]
\addplot+[mark options={fill=blue},fill=blue!40!,draw=blue,  thick]
coordinates
{(17,.7627) (36,.7627) (55,.7627) (74,.7627) (93,.7627) (112,.7627) (131,.7627)}
\closedcycle;%
\addlegendentry{Clean Data}

\addplot+[mark options={fill=green},fill=pink,draw=green,  thick] 
coordinates
 {(20,.5197) (39,.5362) (58,.5671) (77,.6406) (96,.6578) (115,.6842) (134,.7305)}
\closedcycle;%
\addlegendentry{PMLM \cite{li2018privacy}, MLPAM \cite{gupta2020mlpam}}

\addplot+[mark options={fill=white},fill=red!60!,draw=red!70!,  thick] 
coordinates
 {(23,.5760) (42,.6271) (61,.6545) (80,.6624) (99,.7094) (118,.7210) (137,.7477)}
\closedcycle;%
\addlegendentry{DPEL \cite{LI202134}}

\addplot+[mark options={fill=teal},fill=teal,draw=teal,  thick] 
coordinates
 {(26,.6800) (45,.6800) (64,.6800) (83,.6800) (102,.6800) (121,.6800) (140,.6800)}
\closedcycle;%
\addlegendentry{\hspace{-1.75cm} PDLM \cite{ma2018pdlm}}

\addplot+[mark options={fill=cyan},fill=cyan,draw=cyan,  thick] 
coordinates
 {(29,.4455) (48,.5305) (67,.5532) (86,.6185) (105,.6386) (124,.6753) (143,.7094)}
\closedcycle;%
\addlegendentry{DA-PMLM }
\end{axis}
\end{tikzpicture}
            \caption{Balance}
        \end{subfigure}
                  \hfill
\begin{subfigure}[t]{0.49\textwidth}
\begin{tikzpicture}[node distance = 1cm,auto,scale=.70, transform shape]
\pgfplotsset{every axis y label/.append style={rotate=180,yshift=10.5cm}}
\begin{axis}[
      axis on top=false,
      xmin=15, xmax=145,
      ymin=0, ymax=1,
      xtick={22,42,62,80,100,120,138},
      xticklabels={0.01,0.05,0.1,0.5,1.0,1.5,2.0},
        ycomb,
        ylabel near ticks, yticklabel pos=left,
      ylabel={Precision},
      xlabel={Privacy Budget},
      legend style={at={(0.5,-0.18)},
      anchor=north,legend columns=2},
      ymajorgrids=true,
      grid style=dashed,
          ]
\addplot+[mark options={fill=blue},fill=blue!40!,draw=blue,  thick]
coordinates
{(17,.9032) (36,.9032) (55,.9032) (74,.9032) (93,.9032) (112,.9032) (131,.9032)}
\closedcycle;%
\addlegendentry{Clean Data}

\addplot+[mark options={fill=green},fill=pink,draw=green,  thick] 
coordinates
 {(20,.3425) (39,.3817) (58,.3767) (77,.4302) (96,.5801) (115,.5953) (134,.7479)}
\closedcycle;%
\addlegendentry{PMLM \cite{li2018privacy}, MLPAM \cite{gupta2020mlpam}}

\addplot+[mark options={fill=white},fill=red!60!,draw=red!70!,  thick] 
coordinates
 {(23,.3587) (42,.3780) (61,.4013) (80,.5411) (99,.6188) (118,.6444) (137,.7137)}
\closedcycle;%
\addlegendentry{DPEL \cite{LI202134}}

\addplot+[mark options={fill=teal},fill=teal,draw=teal,  thick] 
coordinates
 {(26,.6666) (45,.6666) (64,.6666) (83,.6666) (102,.6666) (121,.6666) (140,.6666)}
\closedcycle;%
\addlegendentry{\hspace{-1.75cm} PDLM \cite{ma2018pdlm}}

\addplot+[mark options={fill=cyan},fill=cyan,draw=cyan,  thick] 
coordinates
 {(29,.3157) (48,.3321) (67,.3410) (86,.3691) (105,.4910) (124,.5911) (143,.7020)}
\closedcycle;%
\addlegendentry{DA-PMLM }
\end{axis}
\end{tikzpicture}
            \caption{Nursery}
        \end{subfigure}
\caption{Precision of \textit{CM} in DA-PMLM}
\end{figure} 
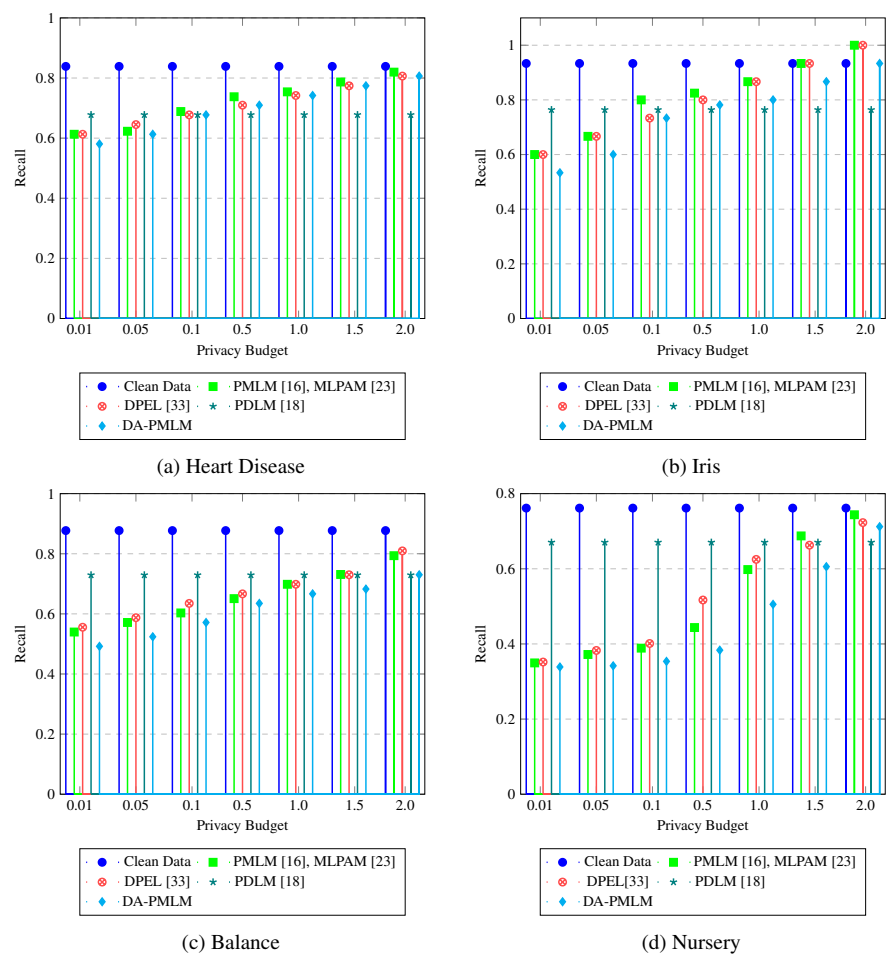
\begin{figure}[!htbp]
\centering
\begin{subfigure}[t]{0.49\textwidth}
\begin{tikzpicture}[node distance = 1cm,auto,scale=.70, transform shape]
\pgfplotsset{every axis y label/.append style={rotate=180,yshift=10.5cm}}
\begin{axis}[
      axis on top=false,
      xmin=15, xmax=145,
      ymin=0, ymax=1,
      xtick={22,42,62,80,100,120,138},
      xticklabels={0.01,0.05,0.1,0.5,1.0,1.5,2.0},
        ycomb,
        ylabel near ticks, yticklabel pos=left,
      ylabel={Recall},
      xlabel={Privacy Budget},
      legend style={at={(0.5,-0.18)},
      anchor=north,legend columns=2},
      ymajorgrids=true,
      grid style=dashed,
          ]
\addplot+[mark options={fill=blue},fill=blue!40!,draw=blue,  thick]
coordinates
{(17,.8387) (36,.8387) (55,.8387) (74,.8387) (93,.8387) (112,.8387) (131,.8387)}
\closedcycle;%
\addlegendentry{Clean Data}

\addplot+[mark options={fill=green},fill=pink,draw=green,  thick] 
coordinates
 {(20,.6129) (39,.6229) (58,.6885) (77,.7377) (96,.7540) (115,.7868) (134,.8196)}
\closedcycle;%
\addlegendentry{PMLM \cite{li2018privacy}, MLPAM \cite{gupta2020mlpam}}

\addplot+[mark options={fill=white},fill=red!60!,draw=red!70!,  thick] 
coordinates
 {(23,.6129) (42,.6451) (61,.6774) (80,.7096) (99,.7419) (118,.7741) (137,.8064)}
\closedcycle;%
\addlegendentry{DPEL \cite{LI202134}}

\addplot+[mark options={fill=teal},fill=teal,draw=teal,  thick] 
coordinates
 {(26,.6777) (45,.6777) (64,.6777) (83,.6777) (102,.6777) (121,.6777) (140,.6777)}
\closedcycle;%
\addlegendentry{\hspace{-1.75cm} PDLM \cite{ma2018pdlm}}

\addplot+[mark options={fill=cyan},fill=cyan,draw=cyan,  thick] 
coordinates
 {(29,.5806) (48,.6129) (67,.6774) (86,.7096) (105,.7419) (124,.7741) (143,.8064)}
\closedcycle;%
\addlegendentry{DA-PMLM }
\end{axis}
\end{tikzpicture}
            \caption{Heart Disease}
        \end{subfigure}
                  \hfill
\begin{subfigure}[t]{0.49\textwidth}
\begin{tikzpicture}[node distance = 1cm,auto,scale=.70, transform shape]
\pgfplotsset{every axis y label/.append style={rotate=180,yshift=10.5cm}}
\begin{axis}[
      axis on top=false,
      xmin=15, xmax=145,
      ymin=0, ymax=1.1,
      xtick={22,42,62,80,100,120,138},
      xticklabels={0.01,0.05,0.1,0.5,1.0,1.5,2.0},
        ycomb,
        ylabel near ticks, yticklabel pos=left,
      ylabel={Recall},
      xlabel={Privacy Budget},
      legend style={at={(0.5,-0.18)},
      anchor=north,legend columns=2},
      ymajorgrids=true,
      grid style=dashed,
          ]
\addplot+[mark options={fill=blue},fill=blue!40!,draw=blue,  thick]
coordinates
{(17,.9333) (36,.9333) (55,.9333) (74,.9333) (93,.9333) (112,.9333) (131,.9333)}
\closedcycle;%
\addlegendentry{Clean Data}

\addplot+[mark options={fill=green},fill=pink,draw=green,  thick] 
coordinates
 {(20,.6000) (39,.6666) (58,.8000) (77,.8247) (96,.8666) (115,.9333) (134,1.0)}
\closedcycle;%
\addlegendentry{PMLM \cite{li2018privacy}, MLPAM \cite{gupta2020mlpam}}

\addplot+[mark options={fill=white},fill=red!60!,draw=red!70!,  thick] 
coordinates
 {(23,.6000) (42,.6666) (61,.7333) (80,.8000) (99,.8666) (118,.93333) (137,1.0)}
\closedcycle;%
\addlegendentry{DPEL \cite{LI202134}}

\addplot+[mark options={fill=teal},fill=teal,draw=teal,  thick] 
coordinates
 {(26,.7638) (45,.7638) (64,.7638) (83,.7638) (102,.7638) (121,.7638) (140,.7638)}
\closedcycle;%
\addlegendentry{\hspace{-1.75cm} PDLM \cite{ma2018pdlm}}

\addplot+[mark options={fill=cyan},fill=cyan,draw=cyan,  thick] 
coordinates
 {(29,.5333) (48,.6000) (67,.7333) (86,.7815) (105,.8000) (124,.8666) (143,.9333)}
\closedcycle;%
\addlegendentry{DA-PMLM }
\end{axis}
\end{tikzpicture}
            \caption{Iris}
        \end{subfigure}
                  \hfill
\begin{subfigure}[t]{0.49\textwidth}
\begin{tikzpicture}[node distance = 1cm,auto,scale=.70, transform shape]
\pgfplotsset{every axis y label/.append style={rotate=180,yshift=10.5cm}}
\begin{axis}[
      axis on top=false,
      xmin=15, xmax=145,
      ymin=0, ymax=1,
      xtick={22,42,62,80,100,120,138},
      xticklabels={0.01,0.05,0.1,0.5,1.0,1.5,2.0},
        ycomb,
        ylabel near ticks, yticklabel pos=left,
      ylabel={Recall},
      xlabel={Privacy Budget},
      legend style={at={(0.5,-0.18)},
      anchor=north,legend columns=2},
      ymajorgrids=true,
      grid style=dashed,
          ]
\addplot+[mark options={fill=blue},fill=blue!40!,draw=blue,  thick]
coordinates
{(17,.8771) (36,.8771) (55,.8771) (74,.8771) (93,.8771) (112,.8771) (131,.8771)}
\closedcycle;%
\addlegendentry{Clean Data}

\addplot+[mark options={fill=green},fill=pink,draw=green,  thick] 
coordinates
 {(20,.5396) (39,.5714) (58,.6031) (77,.6507) (96,.6984) (115,.7311) (134,.7936)}
\closedcycle;%
\addlegendentry{PMLM \cite{li2018privacy}, MLPAM \cite{gupta2020mlpam}}

\addplot+[mark options={fill=white},fill=red!60!,draw=red!70!,  thick] 
coordinates
 {(23,.5555) (42,.5873) (61,.6349) (80,.6666) (99,.6984) (118,.7301) (137,.8095)}
\closedcycle;%
\addlegendentry{DPEL \cite{LI202134}}

\addplot+[mark options={fill=teal},fill=teal,draw=teal,  thick] 
coordinates
 {(26,.7293) (45,.7293) (64,.7293) (83,.7293) (102,.7293) (121,.7293) (140,.7293)}
\closedcycle;%
\addlegendentry{\hspace{-1.75cm} PDLM \cite{ma2018pdlm}}

\addplot+[mark options={fill=cyan},fill=cyan,draw=cyan,  thick] 
coordinates
 {(29,.4920) (48,.5238) (67,.5714) (86,.6349) (105,.6666) (124,.6825) (143,.7301)}
\closedcycle;%
\addlegendentry{DA-PMLM}
\end{axis}
\end{tikzpicture}
            \caption{Balance}
        \end{subfigure}
                  \hfill
\begin{subfigure}[t]{0.49\textwidth}
\begin{tikzpicture}[node distance = 1cm,auto,scale=.70, transform shape]
\pgfplotsset{every axis y label/.append style={rotate=180,yshift=10.5cm}}
\begin{axis}[
      axis on top=false,
      xmin=15, xmax=145,
      ymin=0, ymax=0.8,
      xtick={22,42,62,80,100,120,138},
      xticklabels={0.01,0.05,0.1,0.5,1.0,1.5,2.0},
        ycomb,
        ylabel near ticks, yticklabel pos=left,
      ylabel={Recall},
      xlabel={Privacy Budget},
      legend style={at={(0.5,-0.18)},
      anchor=north,legend columns=2},
      ymajorgrids=true,
      grid style=dashed,
          ]
\addplot+[mark options={fill=blue},fill=blue!40!,draw=blue,  thick]
coordinates
{(17,.7615) (36,.7615) (55,.7615) (74,.7615) (93,.7615) (112,.7615) (131,.7615)}
\closedcycle;%
\addlegendentry{Clean Data}

\addplot+[mark options={fill=green},fill=pink,draw=green,  thick] 
coordinates
 {(20,.3495) (39,.3719) (58,.3888) (77,.4436) (96,.5979) (115,.6875) (134,.7438)}
\closedcycle;%
\addlegendentry{PMLM \cite{li2018privacy}, MLPAM \cite{gupta2020mlpam}}

\addplot+[mark options={fill=white},fill=red!60!,draw=red!70!,  thick] 
coordinates
 {(23,.3518) (42,.3827) (61,.4012) (80,.5169) (99,.6250) (118,.6628) (137,.7229)}
\closedcycle;%
\addlegendentry{DPEL\cite{LI202134}}

\addplot+[mark options={fill=teal},fill=teal,draw=teal,  thick] 
coordinates
 {(26,.6703) (45,.6703) (64,.6703) (83,.6703) (102,.6703) (121,.6703) (140,.6703)}
\closedcycle;%
\addlegendentry{\hspace{-1.75cm} PDLM \cite{ma2018pdlm}}

\addplot+[mark options={fill=cyan},fill=cyan,draw=cyan,  thick] 
coordinates
 {(29,.3387) (48,.3418) (67,.3536) (86,.3834) (105,.5054) (124,.6057) (143,.7121)}
\closedcycle;%
\addlegendentry{DA-PMLM}
\end{axis}
\end{tikzpicture}
            \caption{Nursery}
        \end{subfigure}
\caption{Recall of \textit{CM} in DA-PMLM}
\end{figure} 
\begin{figure}[!htbp]
\centering
\begin{subfigure}[t]{0.49\textwidth}
\begin{tikzpicture}[node distance = 1cm,auto,scale=.70, transform shape]
\pgfplotsset{every axis y label/.append style={rotate=180,yshift=10.5cm}}
\begin{axis}[
      axis on top=false,
      xmin=15, xmax=145,
      ymin=0, ymax=1,
      xtick={22,42,62,80,100,120,138},
      xticklabels={0.01,0.05,0.1,0.5,1.0,1.5,2.0},
        ycomb,
        ylabel near ticks, yticklabel pos=left,
      ylabel={FS},
      xlabel={Privacy Budget},
      legend style={at={(0.5,-0.18)},
      anchor=north,legend columns=2},
      ymajorgrids=true,
      grid style=dashed,
          ]
\addplot+[mark options={fill=blue},fill=blue!40!,draw=blue,  thick]
coordinates
{(17,.8376) (36,.8376) (55,.8376) (74,.8376) (93,.8376) (112,.8376) (131,.8376)}
\closedcycle;%
\addlegendentry{Clean Data}

\addplot+[mark options={fill=green},fill=pink,draw=green,  thick] 
coordinates
 {(20,.6120) (39,.6242) (58,.6896) (77,.7379) (96,.7507) (115,.7876) (134,.8118)}
\closedcycle;%
\addlegendentry{PMLM \cite{li2018privacy}, MLPAM \cite{gupta2020mlpam}}

\addplot+[mark options={fill=white},fill=red!60!,draw=red!70!,  thick] 
coordinates
 {(23,.6137) (42,.6526) (61,.6767) (80,.7784) (99,.7413) (118,.7751) (137,.8085)}
\closedcycle;%
\addlegendentry{DPEL \cite{LI202134}}

\addplot+[mark options={fill=teal},fill=teal,draw=teal,  thick] 
coordinates
 {(26,.6666) (45,.6666) (64,.6666) (83,.6666) (102,.6666) (121,.6666) (140,.6666)}
\closedcycle;%
\addlegendentry{\hspace{-1.75cm} PDLM \cite{ma2018pdlm}}

\addplot+[mark options={fill=cyan},fill=cyan,draw=cyan,  thick] 
coordinates
 {(29,.5832) (48,.6080) (67,.6605) (86,.7019) (105,.7403) (124,.7610) (143,.7963)}
\closedcycle;%
\addlegendentry{DA-PMLM}
\end{axis}
\end{tikzpicture}
            \caption{Heart Disease}
        \end{subfigure}
                  \hfill
\begin{subfigure}[t]{0.49\textwidth}
\begin{tikzpicture}[node distance = 1cm,auto,scale=.70, transform shape]
\pgfplotsset{every axis y label/.append style={rotate=180,yshift=10.5cm}}
\begin{axis}[
      axis on top=false,
      xmin=15, xmax=145,
      ymin=0, ymax=1,
      xtick={22,42,62,80,100,120,138},
      xticklabels={0.01,0.05,0.1,0.5,1.0,1.5,2.0},
        ycomb,
        ylabel near ticks, yticklabel pos=left,
      ylabel={FS},
      xlabel={Privacy Budget},
      legend style={at={(0.5,-0.18)},
      anchor=north,legend columns=2},
      ymajorgrids=true,
      grid style=dashed,
          ]
\addplot+[mark options={fill=blue},fill=blue!40!,draw=blue,  thick]
coordinates
{(17,.9337) (36,.9337) (55,.9337) (74,.9337) (93,.9337) (112,.9337) (131,.9337)}
\closedcycle;%
\addlegendentry{Clean Data}

\addplot+[mark options={fill=green},fill=pink,draw=green,  thick] 
coordinates
 {(20,.5995) (39,.6507) (58,.7458) (77,.7994) (96,.8651) (115,.9344) (134,.9414)}
\closedcycle;%
\addlegendentry{PMLM \cite{li2018privacy}, MLPAM \cite{gupta2020mlpam}}

\addplot+[mark options={fill=white},fill=red!60!,draw=red!70!,  thick] 
coordinates
 {(23,.6071) (42,.6753) (61,.7388) (80,.7948) (99,.8539) (118,.9288) (137,.9344)}
\closedcycle;%
\addlegendentry{DPEL \cite{LI202134}}

\addplot+[mark options={fill=teal},fill=teal,draw=teal,  thick] 
coordinates
 {(26,.7000) (45,.7000) (64,.7000) (83,.7000) (102,.7000) (121,.7000) (140,.7000)}
\closedcycle;%
\addlegendentry{\hspace{-1.75cm} PDLM \cite{ma2018pdlm}}

\addplot+[mark options={fill=cyan},fill=cyan,draw=cyan,  thick] 
coordinates
 {(29,.4666) (48,.5933) (67,.6816) (86,.7836) (105,.8029) (124,.8581) (143,.9333)}
\closedcycle;%
\addlegendentry{DA-PMLM}
\end{axis}
\end{tikzpicture}
            \caption{Iris}
        \end{subfigure}
                  \hfill
\end{figure}
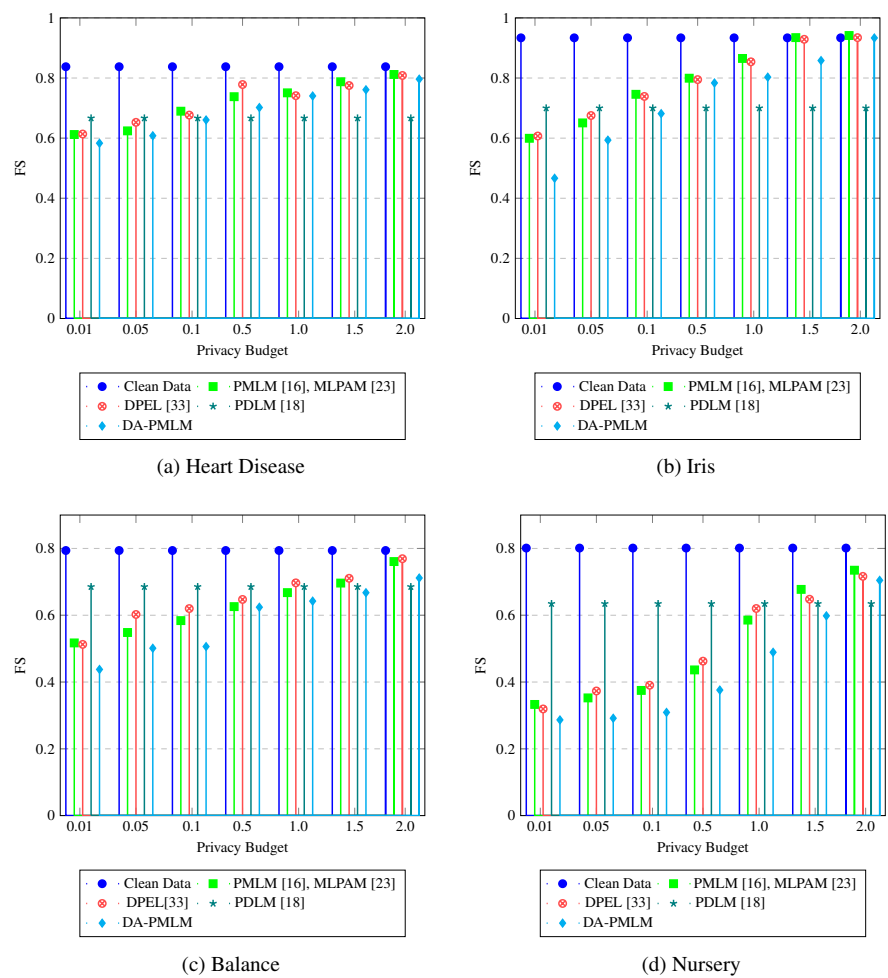
\begin{figure}\setcounter{figure}{6}
\begin{subfigure}[t]{0.49\textwidth}\setcounter{subfigure}{2}
\begin{tikzpicture}[node distance = 1cm,auto,scale=.70, transform shape]
\pgfplotsset{every axis y label/.append style={rotate=180,yshift=10.5cm}}
\begin{axis}[
      axis on top=false,
      xmin=15, xmax=145,
      ymin=0, ymax=0.9,
      xtick={22,42,62,80,100,120,138},
      xticklabels={0.01,0.05,0.1,0.5,1.0,1.5,2.0},
        ycomb,
        ylabel near ticks, yticklabel pos=left,
      ylabel={FS},
      xlabel={Privacy Budget},
      legend style={at={(0.5,-0.18)},
      anchor=north,legend columns=2},
      ymajorgrids=true,
      grid style=dashed,
          ]
\addplot+[mark options={fill=blue},fill=blue!40!,draw=blue,  thick]
coordinates
{(17,.7938) (36,.7938) (55,.7938) (74,.7938) (93,.7938) (112,.7938) (131,.7938)}
\closedcycle;%
\addlegendentry{Clean Data}

\addplot+[mark options={fill=green},fill=pink,draw=green,  thick] 
coordinates
 {(20,.5171) (39,.5483) (58,.5839) (77,.6258) (96,.6677) (115,.6963) (134,.7607)}
\closedcycle;%
\addlegendentry{PMLM \cite{li2018privacy}, MLPAM \cite{gupta2020mlpam}}

\addplot+[mark options={fill=white},fill=red!60!,draw=red!70!,  thick] 
coordinates
 {(23,.5127) (42,.6023) (61,.6200) (80,.6473) (99,.6969) (118,.7106) (137,.7693)}
\closedcycle;%
\addlegendentry{DPEL\cite{LI202134}}

\addplot+[mark options={fill=teal},fill=teal,draw=teal,  thick] 
coordinates
 {(26,.6853) (45,.6853) (64,.6853) (83,.6853) (102,.6853) (121,.6853) (140,.6853)}
\closedcycle;%
\addlegendentry{\hspace{-1.75cm} PDLM \cite{ma2018pdlm}}

\addplot+[mark options={fill=cyan},fill=cyan,draw=cyan,  thick] 
coordinates
 {(29,.4377) (48,.5012) (67,.5062) (86,.6240) (105,.6423) (124,.6675) (143,.7115)}
\closedcycle;%
\addlegendentry{DA-PMLM}
\end{axis}
\end{tikzpicture}
            \caption{Balance}
        \end{subfigure}
                  \hfill
\begin{subfigure}[t]{0.49\textwidth}
\begin{tikzpicture}[node distance = 1cm,auto,scale=.70, transform shape]
\pgfplotsset{every axis y label/.append style={rotate=180,yshift=10.5cm}}
\begin{axis}[
      axis on top=false,
      xmin=15, xmax=145,
      ymin=0, ymax=0.9,
      xtick={22,42,62,80,100,120,138},
      xticklabels={0.01,0.05,0.1,0.5,1.0,1.5,2.0},
        ycomb,
        ylabel near ticks, yticklabel pos=left,
      ylabel={FS},
      xlabel={Privacy Budget},
      legend style={at={(0.5,-0.18)},
      anchor=north,legend columns=2},
      ymajorgrids=true,
      grid style=dashed,
          ]
\addplot+[mark options={fill=blue},fill=blue!40!,draw=blue,  thick]
coordinates
{(17,.8011) (36,.8011) (55,.8011) (74,.8011) (93,.8011) (112,.8011) (131,.8011)}
\closedcycle;%
\addlegendentry{Clean Data}

\addplot+[mark options={fill=green},fill=pink,draw=green,  thick] 
coordinates
 {(20,.3329) (39,.3525) (58,.3749) (77,.4363) (96,.5855) (115,.6774) (134,.7348)}
\closedcycle;%
\addlegendentry{PMLM \cite{li2018privacy}, MLPAM \cite{gupta2020mlpam}}

\addplot+[mark options={fill=white},fill=red!60!,draw=red!70!,  thick] 
coordinates
 {(23,.3194) (42,.3732) (61,.3904) (80,.4625) (99,.6202) (118,.6479) (137,.7164)}
\closedcycle;%
\addlegendentry{DPEL \cite{LI202134}}

\addplot+[mark options={fill=teal},fill=teal,draw=teal,  thick] 
coordinates
 {(26,.6345) (45,.6345) (64,.6345) (83,.6345) (102,.6345) (121,.6345) (140,.6345)}
\closedcycle;%
\addlegendentry{\hspace{-1.75cm} PDLM \cite{ma2018pdlm}}

\addplot+[mark options={fill=cyan},fill=cyan,draw=cyan,  thick] 
coordinates
 {(29,.2866) (48,.2916) (67,.3091) (86,.3760) (105,.4891) (124,.5982) (143,.7045)}
\closedcycle;%
\addlegendentry{DA-PMLM}
\end{axis}
\end{tikzpicture}
            \caption{Nursery}
        \end{subfigure}
\caption{F1-Score of \textit{CM} in DA-PMLM}
\end{figure}
\subsection{Comparison}
The results of the experiments are compared to clean data, PMLM \cite{li2018privacy}, MLPAM \cite{gupta2020mlpam}, DPEL \cite{LI202134}, and PDLM \cite{ma2018pdlm}, which are also deployed on the same platform (Figs. 4(a)-(d) to 7(a)-(d)). The parameters (\textit{CA}, \textit{P}, \textit{R}, and \textit{FS}) results for DA-PMLM are less than the results of Clean, PMLM \cite{li2018privacy}, MLPAM \cite{gupta2020mlpam}, and DPEL \cite{LI202134} in all the cases because of the noise addition in the data and classification model. 
From Table 5, it is observed that the highest difference for \textit{CA} between DA-PMLM and DPEL is 13.35\% on the Nursey dataset with 0.5 privacy budget. The lowest difference is found 0.0\% on the Heart Disease dataset with 0.1, 0.5, 1.0, 1.5, and 2.0 privacy budget, and on the Iris dataset with 0.1 privacy budget. Besides, the maximum gap for \textit{P} is 17.16\% on the Nursery dataset with 0.5 privacy budget, but the lowest differences are found 0.11\% on the Heart Disease dataset with 2.0  privacy budget. The maximum decrement for \textit{R} of DA-PMLM is 13.35\% from DPEL on the Nursey dataset with 0.5 privacy budget. In contrast, the smallest decrement is 0.0\% on the Heart Disease dataset with 0.1, 0.5, 1.0, 1.5, and 2.0 privacy budget, and on the Iris dataset with 0.1 privacy budget. The highest difference for \textit{FS} is 14.05\% on the Iris dataset with 0.01 privacy budget, while the lowest difference for \textit{FS} is 0.1\% on the Heart Disease dataset with 1.0 privacy budget.
\par
Additionally, the parameters (\textit{CA}, \textit{P}, \textit{R}, and \textit{FS}) results for DA-PMLM are less than the results of PMLM and MLPAM in all the cases due to noise addition in the model. The maximum gap for \textit{CA}  is 9.25\% on the Nursey dataset with 1.0 privacy budget, but the smallest gap is found 1.0\% on the Heart Disease with 0.05 privacy budget. Similarly, the highest difference for \textit{P} is 19.33\% on the Iris with 0.01 privacy budget, while the lowest difference is found 0.16\% on the Heart Disease dataset with 0.05 privacy budget. The \textit{R} of DA-PMLM the maximum decrement by 8.72\% from PMLM, MLPAM on the Nursery dataset with 1.0 privacy budget, whereas the smallest decrement is 1.0\% on the Heart Disease dataset with 0.05 privacy budget. The highest difference for \textit{FS} is 13.29\% on the Iris dataset with 0.05 privacy budget, but the lowest difference for \textit{FS} is 0.18\% on the Balance dataset with 0.5 privacy budget.
\begin{table}[!htbp]
\caption{Accuracy, precision, recall, and f1-score; DA-PMLM v/s state-of-the-art works \cite{li2018privacy}, \cite{gupta2020mlpam}, \cite{LI202134} and the baseline schemes}
\label{tab:11}      
\begin{adjustbox}{width=1.0\textwidth} 
\begin{tabular}{ccccc|ccc|ccc|ccc}
\hline
Dataset & Epsilon & \multicolumn{12}{ c }{ Decrements in the value of parameters}  \\ \cline{3-14}
& & \multicolumn{3}{ c }{Accuracy} & \multicolumn{3}{ c }{Precision} & \multicolumn{3}{ c }{Recall} & \multicolumn{3}{ c }{F1-Score} \\
\cline{3-14}
& & Clean data & \cite{li2018privacy}, \cite{gupta2020mlpam} & \cite{LI202134} & Clean data & \cite{li2018privacy}, \cite{gupta2020mlpam} & \cite{LI202134} & Clean data & \cite{li2018privacy}, \cite{gupta2020mlpam} & \cite{LI202134} & Clean data & \cite{li2018privacy}, \cite{gupta2020mlpam} & \cite{LI202134} \\ \hline 
 
& 0.01 & 25.81	& 3.23 & 3.23 & 24.37	& 1.97	& 2.29	& 25.81	& 3.23	& 3.23	& 25.44	& 2.88	& 3.05 \\
& 0.05 & 22.58	& 1	& 3.22	& 21.62	& 0.16	& 4.11	& 22.58	& 1	& 3.22	& 22.96	& 1.62	& 4.46 \\
Heart  & 0.1 & 16.13 & 1.11	& 0	& 17.21	& 4.08	& 1.34	& 16.13	& 1.11	& 0	& 17.71	& 2.91	& 1.62 \\
Disease & 0.5 & 12.91 & 2.81	& 0	& 13.17	& 3.59	& 1.43	& 12.91	& 2.81	& 0	& 13.57	& 3.6	& 0.65 \\
& 1.0 & 9.68 & 1.21	& 0	& 9.55	& 1.29	& 0.21	& 9.68	& 1.21	& 0	& 9.73	& 1.04	& 0.1 \\
& 1.5 & 6.46 & 1.27	& 0	& 6.27	& 3.16	& 1.94	& 6.46	& 1.27	& 0	& 7.66	& 2.66	& 1.41 \\
& 2.0 & 3.23 & 1.32	& 0	& 2.32	& 1.07	& 0.11	& 3.23	& 1.32	& 0	& 4.13	& 2.55	& 1.22 \\ \hline
& 0.01 &  40 & 6.67	& 6.67	& 51.11	& 19.33	& 14.64	& 40 & 6.67	& 6.67	& 46.71	& 13.29	& 14.05 \\
& 0.05 & 33.33 & 6.66 & 6.66 & 36.67 & 19.08 & 11.23 & 33.33 & 6.66	& 6.66 & 34.04 & 5.74 & 8.2 \\
& 0.1 & 20 & 6.67 & 0 & 30.74 & 14.63 & 15.34 & 20	& 6.67	& 0	& 25.21	& 6.42	& 5.72 \\
Iris & 0.5 & 15.18	& 4.32	& 1.85	& 14.45	& 3.56	& 6.67	& 15.18	& 4.32	& 1.85	& 15.01	& 1.58	& 1.12 \\
& 1.0 & 13.33 & 6.66 & 6.66	& 12.53	& 5.59	& 8.09	& 13.33	& 6.66	& 6.66	& 13.08	& 6.22	& 5.1 \\
& 1.5 & 6.67 & 6.67	& 6.67	& 6.83	& 7.05	& 6.55	& 6.67	& 6.67	& 6.67	& 7.56	& 7.63	& 7.07 \\
& 2.0 & 0 & 6.67 & 6.67	& 0	& 2.22	& 0.22	& 0	& 6.67	& 6.67	& 0.04	& 0.81	& 0.11 \\ \hline
& 0.01 & 36.51	& 4.76	& 6.35	& 31.72	& 7.42	& 13.05	& 38.51	& 4.76	& 6.35	& 35.61	& 7.94	& 7.5 \\
& 0.05 & 33.33	& 4.76	& 6.35	& 22.65	& 0.57	& 9.09	& 35.33	& 4.76	& 6.35	& 29.26	& 4.71	& 10.11 \\
& 0.1 & 28.57 & 3.17 & 6.35	& 20.95	& 1.39	& 10.13	& 30.57	& 2.99	& 6.35	& 28.76	& 7.77	& 11.38 \\
Balance & 0.5 & 22.22 & 1.58 & 3.17	& 14.42	& 2.21	& 4.39	& 24.22	& 1.58	& 3.17	& 16.98	& 0.18	& 2.33 \\
& 1.0 & 19.05 & 3.18 & 3.18	& 12.41	& 1.92	& 7.08	& 21.05	& 3.18	& 3.18	& 15.15	& 2.54	& 5.46 \\
& 1.5 & 17.46 & 4.76 & 4.76	& 8.74	& 0.89	& 4.57	& 19.46	& 4.86	& 4.76	& 12.63	& 2.88	& 4.31 \\
& 2.0 & 12.7 & 6.35	& 7.94 &5.33 & 2.11	& 3.83	& 14.7	& 6.35	& 7.94	& 8.23	& 4.92	& 5.78 \\ \hline
& 0.01 & 42.28 & 1.08 & 1.31 & 58.75 & 2.68	& 4.3 & 42.28	& 1.08	& 1.31	& 51.45	& 4.63	& 3.28 \\
& 0.05 & 41.97 & 3.01	& 4.09	& 57.11	& 4.96	& 4.59	& 41.97	& 3.01	& 4.09	& 50.95	& 6.09	& 8.16 \\
& 0.1 & 40.89 & 3.62 & 4.86	& 56.22	& 3.57	& 6.03	& 40.89	& 3.62	& 4.86	& 49.2	& 6.58	& 8.13 \\
Nursery & 0.5 & 37.81 & 6.02 & 13.35 & 53.37 & 6.07	& 17.16	& 37.81	& 6.02	& 13.35	& 42.51	& 6.03 & 8.65 \\
& 1.0 & 25.61 & 9.25 & 11.96 & 41.22 & 8.91	& 12.78	& 25.61	& 9.25 & 11.96	& 31.2	& 9.64	& 13.11 \\
& 1.5 & 15.58 & 8.18 & 5.71	& 31.21	& 10.42	& 5.33	& 15.58	& 8.18	& 5.71	& 20.29	& 7.92	& 4.97 \\
& 2.0 & 4.94 & 3.17	& 1.08	& 20.12	& 4.59	& 1.17	& 4.94	& 3.17	& 1.08	& 9.66	& 3.03	& 1.19 \\ \hline
\end{tabular}
\end{adjustbox}
\end{table}
\begin{table}[!htbp]
\centering
\caption{Improvement in the values of \textit{CA}, \textit{P}, \textit{R}, and \textit{FS} of DA-PMLM vs PDLM \cite{ma2018pdlm}}
\label{tab:12}      
\begin{tabular}{ccc|c|c|c}
\hline
Dataset & Epsilon & \multicolumn{4}{ c }{Increments in the value of parameters}  \\ \cline{3-6}
& & \textit{CA} & \textit{P} & \textit{R} & \textit{FS} \\
\hline 
 
& 0.01 & -9.71	& -8.33 & -9.71 & -8.34	 \\
& 0.05 & -6.48	& -5.58	& -6.48	& -5.86	 \\
Heart  & 0.1 & -0.03 & -1.17 & -0.03 & -0.61 \\
Disease & 0.5 & 3.19 & 2.87 & 3.19 &3.53  \\
& 1.0 & 6.42 & 6.49	& 6.42 & 7.37	 \\
& 1.5 & 9.64 & 9.77 & 9.64 & 9.44	 \\
& 2.0 & 12.87 & 13.72 & 12.87 & 12.97	 \\ \hline
& 0.01 & -23.05  & -30.95 & -23.05	& -23.34 \\
& 0.05 & -16.38 & -16.51 & -16.38 & -10.67  \\
& 0.1 & -3.05 & -10.58 & -3.05 & -1.84  \\
Iris & 0.5 & 1.77 & 5.71 & 1.77	& 8.36	 \\
& 1.0 & 3.62 & 7.63 & 3.62	& 	10.29 \\
& 1.5 & 10.28 & 13.33 & 10.28 & 15.81	\\
& 2.0 & 16.95 & 20.16 & 16.95 & 23.33	 \\ \hline
& 0.01 & -23.73	& -23.45 & -23.73	& 	-24.76	 \\
& 0.05 & -20.55	& -14.95 & -20.55	& -18.41	 \\
& 0.1 & -15.79 & -12.68 & -15.79	& -17.91	 \\
Balance & 0.5 & -9.44 & -6.15 &	-9.44 & -6.13 \\
& 1.0 & -6.27 & -4.14 & -6.27	& -4.3	\\
& 1.5 & -4.68 & -0.47 & -4.68	& 	-1.78	 \\
& 2.0 & 0.08 & 2.94	& 0.08 & 2.62 \\ \hline
& 0.01 & -33.16 & -35.09 & -33.16 & -34.79 \\
& 0.05 & -32.85 & -33.45	& -32.85	& 	-34.29	 \\
& 0.1 & -31.77 & -32.56 & -31.77	& -32.54 	 \\
Nursery & 0.5 & -28.69  & -29.71 & 	-28.69 & -25.85  \\
& 1.0 & -16.49 & -17.56 & -16.49 &  -14.54 \\
& 1.5 & -6.46 & -7.55  & -6.46	& 	-3.63 \\
& 2.0 & 4.18 & 	3.54	&  4.18	& 7.00	\\ \hline
\end{tabular}
\end{table}
\par 
Likewise,  DA-PMLM outperforms PDLM \cite{ma2018pdlm} because the proposed model's efficiency is increased by performing fewer calculations on synthetic data instead of ciphertext. Table 6 shows that the improvement ranges for \textit{CA}, \textit{P}, \textit{R}, and \textit{FS} are 0.08\% to 16.95\%, 2.87\% to 20.16\%, 0.08\% to 16.95\%, and 2.62\% to 23.33\% over all datasets.
\par
Moreover, the parameters (\textit{CA}, \textit{P}, \textit{R}, and \textit{FS}) results for DA-PMLM are less than the results of clean data in all the cases due to noise addition. The maximum gap for \textit{CA}  is 42.28\% on the Nursey dataset with 0.01 privacy budget, but the lowest gap is found 0.0\% on the Iris dataset with 2.0 privacy budget. Similarly, the highest difference for \textit{P} is 58.75\% on the Nursey dataset with 0.01 privacy budget, while the lowest difference is found 0.0\% on the Iris Disease with 2.0 privacy budget. The maximum decrement for \textit{R} of DA-PMLM is 42.28\% from clean data on the Nursey dataset with 0.01 privacy budget, whereas the smallest decrement is 0.0\% on the Iris dataset with 2.0 privacy budget. The highest difference for \textit{FS} is 51.45\% on the Nursey dataset with 0.01 privacy budget, but the lowest difference for \textit{FS} is 0.04\% on the Iris dataset with 2.0 privacy budget. 
However, The parameters (\textit{CA}, \textit{P}, \textit{R}, and \textit{FS}) results for DA-PMLM are less or equal and provide more protection than the clean data, PMLM, MLPAM, DPEL, and PDLM. The proposed model acquires a trade-off between accuracy loss and privacy improvement when the noise is added using individual privacy constraints.
\par
Furthermore, a comprehensive feature analysis along with a comparison of DA-PMLM against the state-of-the-art works \cite{li2018privacy}, \cite{gupta2020mlpam}, \cite{LI202134}, \cite{ma2018pdlm} is performed. From Table 7, it is observed that DA-PMLM is the only model that preserves the privacy of data and classification model. It also permits several untrusted entities, data owners, request users, and classifier owners to participate. DA-PMLM works on the dynamic privacy budget. In DA-PMLM, the use of the privacy budget is according to the level of protection required. Therefore, the performance of DA-PMLM outperforms than the existing schemes \cite{li2018privacy}, \cite{gupta2020mlpam}, \cite{LI202134}, \cite{ma2018pdlm}.
\begin{table}[!htbp] 
\centering
\caption{Comparative analysis of relevant schemes}
\resizebox{\textwidth}{!}{
\begin{tabular}{llllcccc}
\hline
Schemes & UE & DO  & RU & CO & Privacy & Data  & Model    \\     
&   &  &  &  & Budget & Privacy & Privacy   \\ \hline     
PMLM \cite{li2018privacy} & multiple & multiple & no & $\times$ & static & $\checkmark$ &  $\times$    \\ \hline
MLPAM \cite{gupta2020mlpam}& multiple & multiple & multiple & $\times$ & static & $\checkmark$ &  $\times$   \\ \hline
DPEL \cite{LI202134} & single & single & no  & $\times$ & dynamic & $\times$ & $\checkmark$  \\ \hline
PDLM \cite{ma2018pdlm} & single & multiple & no & $\times$ & no & $\checkmark$ & $\times$  \\ \hline
DA-PMLM & multiple & multiple & multiple &  $\checkmark$ & dynamic & $\checkmark$ & $\checkmark$  \\ \hline
\noalign{\smallskip}
\end{tabular}}
\end{table}
\subsection{Statistical Analysis}
Statistical analysis is used to validate the \textit{CA}, \textit{P}, \textit{R}, and \textit{FS} of the model. In this context, the non-parametric test is applied to the dataset that is not normally distributed. The null hypothesis states that the acquired results from different methods are statistically identical in the Wilcoxon signed-rank test. This test compares the performance of DA-PMLM to that of the existing DPEL \cite{LI202134}, PMLM \cite{li2018privacy}, MLPAM \cite{gupta2020mlpam}, and PDLM \cite{ma2018pdlm} models. The test is run on the dataset with a significance level \cite{kumar2020biphase} (p-value) of 0.05 to determine the importance of classifying parameters. Table 8 demonstrates the results of the test statistics.
\begin{table}[!htbp]
\centering
\caption{Wilcoxon test statistics (p-value is 0.05)}
\label{table:21}
\resizebox{\textwidth}{!}{
\begin{tabular}{llllllll}
\hline
Dataset & Classification & \multicolumn{2}{ c }{Comparison of} & \multicolumn{2}{ c }{Comparison of DA-PMLM,} & \multicolumn{2}{ c }{Comparison of} \\  
&  Parameters &\multicolumn{2}{ c }{DA-PMLM \& DPEL \cite{LI202134}} & \multicolumn{2}{ c }{PMLM \cite{li2018privacy} \& MLPAM \cite{gupta2020mlpam}} & \multicolumn{2}{ c }{DA-PMLM \& PDLM \cite{ma2018pdlm}}   \\    \cline{3-8}
& & p-value & Result & p-value & Result & p-value & Result \\ \hline 
& Accuracy  & 0.180 & Accepted & 0.018 & Rejected & 0.612 & Accepted  \\ 
Heart & Precision & 0.018 & Rejected & 0.018 & Rejected & 0.398 & Accepted \\ 
Disease & Recall & 0.180 & Accepted & 0.018 & Rejected & 0.612 & Accepted \\ 
& F1-score & 0.018 & Rejected & 0.018 & Rejected & 0.398 & Accepted  \\ \hline
& Accuracy  & 0.026 & Rejected & 0.016 & Rejected & 1.000 & Accepted \\ 
Iris & Precision  & 0.018 & Rejected & 0.018 & Rejected & 0.866 & Accepted  \\ 
 & Recall  & 0.026 & Rejected & 0.016 & Rejected & 1.000 & Accepted  \\ 
& F1-score  & 0.018 & Rejected  & 0.018 & Rejected & 0.735 & Accepted \\ \hline
& Accuracy  & 0.017 & Rejected & 0.018 & Rejected & 0.028 & Rejected\\ 
Balance & Precision  & 0.018 & Rejected & 0.018 & Rejected & 0.043 & Rejected \\ 
Scale & Recall  & 0.017 & Rejected & 0.017 & Rejected & 0.028 & Rejected\\ 
& F1-score  & 0.018 & Rejected  & 0.018 & Rejected & 0.043 & Rejected \\ \hline
& Accuracy & 0.018 & Rejected & 0.018 & Rejected  & 0.028 & Rejected\\ 
Nursery & Precision  & 0.018 & Rejected & 0.018 & Rejected & 0.028 & Rejected\\ 
 & Recall  & 0.018 & Rejected & 0.018 & Rejected & 0.028 & Rejected \\ 
& F1-score  & 0.018 & Rejected & 0.018 & Rejected & 0.043 & Rejected \\ \hline
\noalign{\smallskip}
\end{tabular}}
\end{table}
\par
While comparing DA-PMLM and DPEL, it is observed that the null hypothesis for \textit{P}, and \textit{FS} is rejected but accepted for \textit{CA}, and \textit{R} on the Heart Disease dataset. The null hypothesis is rejected for \textit{CA}, \textit{P}, \textit{R}, and \textit{FS} on the Iris, Balance Scale, and Nursery dataset, because their p-values are less than 0.05. Similarly, comparing DA-PMLM, PMLM, and MLPAM, the null hypothesis is rejected for \textit{CA}, \textit{P}, \textit{R}, and \textit{FS} on the Heart Disease, Iris, Balance Scale, and Nursery datasets. Moreover, comparing DA-PMLM and PDLM, the null hypothesis is accepted for \textit{CA}, \textit{P}, \textit{R}, and \textit{FS} on the Heart Disease, Iris datasets, but rejected for Balance Scale, and Nursery datasets. Based on the achieved statistics results, DA-PMLM improves \textit{CA}, \textit{P}, \textit{R}, \textit{FS}, which demonstrates the superiority of the proposed model.
\section{Conclusion}
This paper proposed a novel model named DA-PMLM that protects the sensitive data and classification model outsourced by multiple owners in a real cloud environment. DA-PMLM permits the various data owners to store and analyze their data and the classifier owner to share the model to handle the multiple requests in the cloud. In this work, data owners inject varying statistical noise to sensitive data according to their queries as well as classifier owner adds the statistical noise into the model to preserve privacy. The experiments have been performed, and the results demonstrate that the parameter's values of classification of DA-PMLM (accuracy, precision, recall, and f1-score) are degraded by increased privacy (i.e., lower $\epsilon$). However, the results also illustrate that more elevated $\epsilon$ can enhance the parameter's classification values of DA-PMLM. The performance of DA-PMLM on the prominent data sets is more secure, efficient, and optimal than existing works in this regard. As part of our future work, we will devise a more efficient privacy-preserving mechanism that reduces performance degradation.
\bibliographystyle{spbasic}      
\bibliography{reference.bib}

\begin{thebibliography}{10}

\bibitem{iwasa2020development}
Daiji Iwasa, Teruaki Hayashi, and Yukio Ohsawa.
\newblock Development and evaluation of a new platform for accelerating
  cross-domain data exchange and cooperation.
\newblock {\em New Generation Computing}, 38(1):65--96, 2020.

\bibitem{xu2018secure}
Shengmin Xu, Guomin Yang, Yi~Mu, and Robert~H Deng.
\newblock Secure fine-grained access control and data sharing for dynamic
  groups in the cloud.
\newblock {\em IEEE Transactions on Information Forensics and Security},
  13(8):2101--2113, 2018.

\bibitem{stergiou2017efficient}
Christos Stergiou and Kostas~E Psannis.
\newblock Efficient and secure big data delivery in cloud computing.
\newblock {\em Multimedia Tools and Applications}, 76(21):22803--22822, 2017.

\bibitem{sarker2021machine}
Iqbal~H Sarker.
\newblock Machine learning: Algorithms, real-world applications and research
  directions.
\newblock {\em SN Computer Science}, 2(3):1--21, 2021.

\bibitem{batouche2021comprehensive}
Ali Batouche and Hamid Jahankhani.
\newblock A comprehensive approach to android malware detection using machine
  learning.
\newblock In {\em Information Security Technologies for Controlling Pandemics},
  pages 171--212. Springer, 2021.

\bibitem{ghorbel2017privacy}
Amal Ghorbel, Mahmoud Ghorbel, and Mohamed Jmaiel.
\newblock Privacy in cloud computing environments: a survey and research
  challenges.
\newblock {\em The Journal of Supercomputing}, 73(6):2763--2800, 2017.

\bibitem{ali2015security}
Mazhar Ali, Samee~U Khan, and Athanasios~V Vasilakos.
\newblock Security in cloud computing: Opportunities and challenges.
\newblock {\em Information sciences}, 305:357--383, 2015.

\bibitem{shen2018enabling}
Wenting Shen, Jing Qin, Jia Yu, Rong Hao, and Jiankun Hu.
\newblock Enabling identity-based integrity auditing and data sharing with
  sensitive information hiding for secure cloud storage.
\newblock {\em IEEE Transactions on Information Forensics and Security},
  14(2):331--346, 2018.

\bibitem{cisco}
Cisco.
\newblock {Cisco Secure}.
\newblock
  \url{https://www.cisco.com/c/dam/en_us/about/doing_business/trust-center/docs/cisco-cybersecurity-series-2021-cps.pdf},
  2021.
\newblock [Online; accessed 2021].

\bibitem{wei2014security}
Lifei Wei, Haojin Zhu, Zhenfu Cao, Xiaolei Dong, Weiwei Jia, Yunlu Chen, and
  Athanasios~V Vasilakos.
\newblock Security and privacy for storage and computation in cloud computing.
\newblock {\em Information sciences}, 258:371--386, 2014.

\bibitem{peyvandi2021computer}
Amirhossein Peyvandi, Babak Majidi, Soodeh Peyvandi, and Jagdish Patra.
\newblock Computer-aided-diagnosis as a service on decentralized medical cloud
  for efficient and rapid emergency response intelligence.
\newblock {\em New Generation Computing}, pages 1--24, 2021.

\bibitem{dwork2006calibrating}
Cynthia Dwork, Frank McSherry, Kobbi Nissim, and Adam Smith.
\newblock Calibrating noise to sensitivity in private data analysis.
\newblock In {\em Theory of cryptography conference}, pages 265--284. Springer,
  2006.

\bibitem{zhao2019novel}
Xiaodong Zhao, Yulan Dong, and Dechang Pi.
\newblock Novel trajectory data publishing method under differential privacy.
\newblock {\em Expert Systems with Applications}, 138:112791, 2019.

\bibitem{fang2020local}
Xianjin Fang, Qingkui Zeng, and Gaoming Yang.
\newblock Local differential privacy for data streams.
\newblock In {\em International Conference on Security and Privacy in Digital
  Economy}, pages 143--160. Springer, 2020.

\bibitem{thilakanathan2014secure}
Danan Thilakanathan, Shiping Chen, Surya Nepal, and Rafael~A Calvo.
\newblock Secure data sharing in the cloud.
\newblock In {\em Security, privacy and trust in cloud systems}, pages 45--72.
  Springer, 2014.

\bibitem{li2018privacy}
Ping Li, Tong Li, Heng Ye, Jin Li, Xiaofeng Chen, and Yang Xiang.
\newblock Privacy-preserving machine learning with multiple data providers.
\newblock {\em Future Generation Computer Systems}, 87:341--350, 2018.

\bibitem{wang2018privacy}
Zhibo Wang, Xiaoyi Pang, Yahong Chen, Huajie Shao, Qian Wang, Libing Wu,
  Honglong Chen, and Hairong Qi.
\newblock Privacy-preserving crowd-sourced statistical data publishing with an
  untrusted server.
\newblock {\em IEEE Transactions on Mobile Computing}, 18(6):1356--1367, 2018.

\bibitem{ma2018pdlm}
Xindi Ma, Jianfeng Ma, Hui Li, Qi~Jiang, and Sheng Gao.
\newblock Pdlm: Privacy-preserving deep learning model on cloud with multiple
  keys.
\newblock {\em IEEE Transactions on Services Computing}, 14(4):1251--1263,
  2018.

\bibitem{hassan2019efficient}
Alzubair Hassan, Rafik Hamza, Hongyang Yan, and Ping Li.
\newblock An efficient outsourced privacy preserving machine learning scheme
  with public verifiability.
\newblock {\em IEEE Access}, 7:146322--146330, 2019.

\bibitem{fan2020privacy}
Weibei Fan, Jing He, Mengjiao Guo, Peng Li, Zhijie Han, and Ruchuan Wang.
\newblock Privacy preserving classification on local differential privacy in
  data centers.
\newblock {\em Journal of Parallel and Distributed Computing}, 135:70--82,
  2020.

\bibitem{wei2020federated}
Kang Wei, Jun Li, Ming Ding, Chuan Ma, Howard~H Yang, Farhad Farokhi, Shi Jin,
  Tony~QS Quek, and H~Vincent Poor.
\newblock Federated learning with differential privacy: Algorithms and
  performance analysis.
\newblock {\em IEEE Transactions on Information Forensics and Security},
  15:3454--3469, 2020.

\bibitem{liu2020adaptive}
Xiaoyuan Liu, Hongwei Li, Guowen Xu, Rongxing Lu, and Miao He.
\newblock Adaptive privacy-preserving federated learning.
\newblock {\em Peer-to-Peer Networking and Applications}, 13(6):2356--2366,
  2020.

\bibitem{gupta2020mlpam}
Ishu Gupta, Rishabh Gupta, Ashutosh~Kumar Singh, and Rajkumar Buyya.
\newblock Mlpam: A machine learning and probabilistic analysis based model for
  preserving security and privacy in cloud environment.
\newblock {\em IEEE Systems Journal}, 2020.

\bibitem{xiong2020real}
Xingxing Xiong, Shubo Liu, Dan Li, Zhaohui Cai, and Xiaoguang Niu.
\newblock Real-time and private spatio-temporal data aggregation with local
  differential privacy.
\newblock {\em Journal of Information Security and Applications}, 55:102633,
  2020.

\bibitem{liu2020local}
Peng Liu, YuanXin Xu, Quan Jiang, Yuwei Tang, Yameng Guo, Li-e Wang, and
  Xianxian Li.
\newblock Local differential privacy for social network publishing.
\newblock {\em Neurocomputing}, 391:273--279, 2020.

\bibitem{sharma2021differential}
Jhilakshi Sharma, Donghyun Kim, Ahyoung Lee, and Daehee Seo.
\newblock On differential privacy-based framework for enhancing user data
  privacy in mobile edge computing environment.
\newblock {\em IEEE Access}, 9:38107--38118, 2021.

\bibitem{li2018differentially}
Tong Li, Jin Li, Zheli Liu, Ping Li, and Chunfu Jia.
\newblock Differentially private naive bayes learning over multiple data
  sources.
\newblock {\em Information Sciences}, 444:89--104, 2018.

\bibitem{GAO201872}
Chong zhi Gao, Qiong Cheng, Pei He, Willy Susilo, and Jin Li.
\newblock Privacy-preserving naive bayes classifiers secure against the
  substitution-then-comparison attack.
\newblock {\em Information Sciences}, 444:72--88, 2018.

\bibitem{li2018privacy1234}
Ping Li, Jin Li, Zhengan Huang, Chong-Zhi Gao, Wen-Bin Chen, and Kai Chen.
\newblock Privacy-preserving outsourced classification in cloud computing.
\newblock {\em Cluster Computing}, 21(1):277--286, 2018.

\bibitem{LIU2018807}
Xiaoqian Liu, Qianmu Li, Tao Li, and Dong Chen.
\newblock Differentially private classification with decision tree ensemble.
\newblock {\em Applied Soft Computing}, 62:807--816, 2018.

\bibitem{xiong2019differential}
Zhili Xiong, Longyuan Li, Junchi Yan, Haiyang Wang, Hao He, and Yaohui Jin.
\newblock Differential privacy with variant-noise for gaussian processes
  classification.
\newblock In {\em Pacific Rim International Conference on Artificial
  Intelligence}, pages 107--119. Springer, 2019.

\bibitem{wang2020differential}
Puyu Wang and Hai Zhang.
\newblock Differential privacy for sparse classification learning.
\newblock {\em Neurocomputing}, 375:91--101, 2020.

\bibitem{LI202134}
Xianxian Li, Jing Liu, Songfeng Liu, and Jinyan Wang.
\newblock Differentially private ensemble learning for classification.
\newblock {\em Neurocomputing}, 430:34--46, 2021.

\bibitem{chen2014correlated}
Rui Chen, Benjamin~CM Fung, S~Yu Philip, and Bipin~C Desai.
\newblock Correlated network data publication via differential privacy.
\newblock {\em The VLDB Journal}, 23(4):653--676, 2014.

\bibitem{mcsherry2009privacy}
Frank~D McSherry.
\newblock Privacy integrated queries: an extensible platform for
  privacy-preserving data analysis.
\newblock In {\em Proceedings of the 2009 ACM SIGMOD International Conference
  on Management of data}, pages 19--30, 2009.

\bibitem{berrar2018bayes}
Daniel Berrar.
\newblock Bayes’ theorem and naive bayes classifier.
\newblock {\em Encyclopedia of Bioinformatics and Computational Biology: ABC of
  Bioinformatics; Elsevier Science Publisher: Amsterdam, The Netherlands},
  pages 403--412, 2018.

\bibitem{song2012automatic}
Qinbao Song, Guangtao Wang, and Chao Wang.
\newblock Automatic recommendation of classification algorithms based on data
  set characteristics.
\newblock {\em Pattern recognition}, 45(7):2672--2689, 2012.

\bibitem{kumar2020biphase}
Jitendra Kumar, Deepika Saxena, Ashutosh~Kumar Singh, and Anand Mohan.
\newblock Biphase adaptive learning-based neural network model for cloud
  datacenter workload forecasting.
\newblock {\em Soft Computing}, pages 1--18, 2020.

\end{thebibliography}
\end{document}